\documentclass[11pt]{article}

\usepackage{amsthm, amsmath}
\usepackage{amssymb}
\usepackage{cancel}

\usepackage[margin=1in, letterpaper]{geometry}
\usepackage{microtype}
\usepackage[T1]{fontenc}
\usepackage{lmodern}
\usepackage{verbatim}
\usepackage{enumitem}
\usepackage{array}
\usepackage{multirow}
\usepackage{afterpage}
\usepackage[usenames,dvipsnames]{xcolor}
\usepackage[colorlinks=true,allcolors=Maroon]{hyperref}
\usepackage{caption} 
\usepackage{subcaption}
\usepackage{setspace}
\usepackage{braket}
\usepackage{graphicx}
\usepackage{algpseudocode}
\usepackage{algorithm}

\algtext*{EndWhile}
\algtext*{EndFor}
\algtext*{EndIf}
\algtext*{EndFunction}

\newtheorem{theorem}{Theorem}[section]
\newtheorem{lemma}[theorem]{Lemma}

\newtheorem{claim}[theorem]{Claim}

\newtheorem{corollary}[theorem]{Corollary}

\newtheorem{definition}[theorem]{Definition}

\begin{document}

\title{
Optimal Hashing-based Time--Space Trade-offs for\\Approximate Near Neighbors%
\footnote{This paper merges two arXiv preprints:~\cite{Laarhoven2015}
  (appeared online on November~24,~2015) and~\cite{ALRW16} (appeared
  online on May~9,~2016), and subsumes both of these articles. An
  extended abstract of this paper appeared in the proceedings of 28th Annual ACM--SIAM Symposium on Discrete Algorithms (SODA~'2017).}\\ }

\author{
Alexandr Andoni\\Columbia
\and
Thijs Laarhoven\\IBM Research Z\"urich
\and
Ilya Razenshteyn\\MIT CSAIL
\and 
Erik Waingarten\\Columbia
}

\maketitle

\begin{abstract}
We show tight upper and lower bounds for time--space trade-offs for
the $c$-Approximate Near Neighbor Search problem.  For the
$d$-dimensional Euclidean space and $n$-point datasets, we develop a
data structure with space $n^{1 + \rho_u + o(1)} + O(dn)$ and query
time $n^{\rho_q + o(1)} + d n^{o(1)}$ for every $\rho_u, \rho_q \geq
0$ with:
\begin{equation}
\label{tradeoff_abstract}
c^2 \sqrt{\rho_q} + (c^2 - 1) \sqrt{\rho_u} = \sqrt{2c^2 - 1}.
\end{equation}

For example, for the approximation $c = 2$ we can achieve:
\begin{itemize}
\item Space $n^{1.77\ldots}$ and query time $n^{o(1)}$,
significantly improving upon known data structures that support very fast queries~\cite{IM, KOR00};
\item Space $n^{1.14\ldots}$ and query time $n^{0.14\ldots}$, matching the optimal data-dependent Locality-Sensitive Hashing (LSH) from~\cite{AR-optimal};
\item Space $n^{1 + o(1)}$ and query time $n^{0.43\ldots}$,
making significant progress in the regime of near-linear space, which is arguably
of the most interest for practice~\cite{lv2007multi}. 
\end{itemize}
This is the first data structure that achieves sublinear query time
and near-linear space for \emph{every} approximation factor $c > 1$,
improving upon~\cite{Kap15-tradeoff}.
The data structure is a culmination of a~long line of work on the problem for all
space regimes; it builds on spherical Locality-Sensitive Filtering~\cite{BDGL16-subRho} and
data-dependent hashing~\cite{AINR-subLSH, AR-optimal}. 

Our matching lower bounds are of two types: conditional and
unconditional. First, we prove tightness of the \emph{whole}
trade-off~\eqref{tradeoff_abstract} in a restricted model of
computation, which captures all known hashing-based approaches. We
then show \emph{unconditional} cell-probe lower bounds for one and two
probes that match~\eqref{tradeoff_abstract} for $\rho_q = 0$, improving upon the best known lower bounds from \cite{PTW10}.
In particular, this is the first space lower bound (for \emph{any}
static data structure) for two probes which is not polynomially
smaller than the corresponding one-probe bound. To show the result for two probes, we
establish and exploit a connection to \emph{locally-decodable codes}.
\end{abstract}

\thispagestyle{empty}
\newpage
\thispagestyle{empty}

{
\small 
\tableofcontents
}

\newpage
\setcounter{page}{1}

\onehalfspacing

\def\colorful{1}

\ifnum\colorful=1
\newcommand{\violet}[1]{{\color{violet}{#1}}}
\newcommand{\orange}[1]{{\color{orange}{#1}}}
\newcommand{\blue}[1]{{{\color{blue}#1}}}
\newcommand{\red}[1]{{\color{red} {#1}}}
\newcommand{\green}[1]{{\color{green} {#1}}}
\newcommand{\pink}[1]{{\color{pink}{#1}}}
\newcommand{\gray}[1]{{\color{gray}{#1}}}

\fi
\ifnum\colorful=0
\newcommand{\violet}[1]{{{#1}}}
\newcommand{\orange}[1]{{{#1}}}
\newcommand{\blue}[1]{{{#1}}}
\newcommand{\red}[1]{{{#1}}}
\newcommand{\green}[1]{{{#1}}}
\newcommand{\gray}[1]{{{#1}}}

\fi
 
\newcommand{\R}{\mathbb{R}}
\newcommand{\E}{\mathbb{E}}
\newcommand{\1}{\mathbf{1}}
\newcommand{\A}{\mathcal{A}}
\newcommand{\Pd}{\mathcal{P}}
\newcommand{\Q}{\mathcal{Q}}
\newcommand{\X}{\mathcal{X}}
\newcommand{\eps}{\varepsilon}
\newcommand{\uth}{{\huge \red{UP TO HERE}}}
\newcommand{\Rbb}{\mathbb{R}}
\newcommand{\Pc}{\mathcal{P}}
\newcommand{\Rc}{\mathcal{R}}
\newcommand{\Lc}{\mathcal{L}}
\newcommand{\Prb}[2]{\underset{#1}{\mathrm{Pr}}\left[#2\right]}

\newcommand{\clevercomplete}{\mathrm{clevercomplete}}

\newcommand{\ignore}[1]{}

\newcommand{\aanote}[1]{}
\newcommand{\irnote}[1]{}
\newcommand{\ewnote}[1]{}
\newcommand{\tlnote}[1]{}

\newcommand{\Qeta}{\eta_q}
\newcommand{\Ueta}{\eta_u}


\section{Introduction}

\subsection{Approximate Near Neighbor problem (ANN)}

The Near Neighbor Search problem (NNS) is a basic and fundamental problem in
computational geometry, defined as follows. We are given a dataset $P$ of
$n$ points from a metric space $(X, d_X)$ and a distance threshold
$r > 0$.  The goal is to preprocess $P$ in order to answer
\emph{near neighbor queries}: given a query point $q \in X$, return a
dataset point $p \in P$ with $d_X(q, p) \leq r$, or report that there is no
such point. The $d$-dimensional Euclidean $(\R^d, \ell_2)$ and Manhattan/Hamming $(\R^d,
\ell_1)$ metric spaces have received the most attention.  Besides
classical applications to similarity search over many types of data
(text, audio, images, etc; see~\cite{NNNIPS} for an overview), NNS has
been also recently used for cryptanalysis~\cite{MO15-nns, L15a, l-svp-15, BDGL16-subRho} and
optimization~\cite{DRT11, HLM15, ZYS16}. 

The performance of an NNS data structure is primarily characterized by
two key metrics:
\begin{itemize}
\item space: the amount of memory a data structure occupies, and
\item query time: the time it takes to answer a query.
\end{itemize}

All known time-efficient data
structures for NNS (e.g., \cite{Cl1, Mei93}) require space
exponential in the dimension $d$, which is prohibitively expensive unless
$d$ is very small.  To overcome this so-called \emph{curse of
dimensionality}, researchers proposed the $(c,r)$-\emph{Approximate}
Near Neighbor Search problem, or $(c,r)$-ANN. In this relaxed version,
we are given a dataset $P$
and a
distance threshold $r > 0$, as well as an approximation factor $c >
1$.  Given a query point $q$ with the promise that there is at least one
data point in $P$ within distance at most $r$ from $q$, the goal is to return a~data
point $p \in P$ within a distance at most $cr$ from $q$.

ANN does allow efficient data structures with
a query time sublinear in~$n$, and only polynomial dependence in $d$
in all parameters \cite{IM, GIM, KOR00, I-thesis, I98, Char, CR04, DIIM, Pan, AI-CACM, TT-spherical, AC-fastJL, AINR-subLSH, Kap15-tradeoff, AR-optimal, Pagh16-deterministic, BDGL16-subRho, ARS17, ANRW17}. In
practice, ANN algorithms are often successful even when one is
interested in \emph{exact} nearest neighbors~\cite{ADIIM, AILSR15}.
We refer the reader to~\cite{HIM12, AI-CACM, Andoni-Thesis} for
a~survey of the theory of ANN, and
\cite{wssj-hsss-14,WLKC15-learningHash} for a more practical
perspective.

In this paper, we obtain tight time--space trade-offs for ANN in hashing-based models. Our
upper bounds are stated in Section~\ref{sec:UBresults}, and the lower
bounds are stated in Section~\ref{sec:LBresults}.  We provide more
background on the problem next. 

\subsection{Locality-Sensitive Hashing (LSH) and beyond}

A classic technique for ANN is \emph{Locality-Sensitive Hashing}
(LSH), introduced in~1998 by Indyk and Motwani~\cite{IM, HIM12}. The
main idea is to use \emph{random space partitions}, 
for which a pair of close points
(at distance at most $r$) is more likely to belong to the same part than a pair of far
points (at distance more than $cr$). Given such a partition, the data
structure splits the dataset $P$ according to the partition, and, given a query,
retrieves all the data points which belong to the same part as the query.
In order to return a near-neighbor with high probability of success, one maintains several partitions
and checks all of them during the query stage. LSH yields data structures with space $O(n^{1 + \rho} + d n)$
and query time $O(d n^{\rho})$, where $\rho$ is the key quantity measuring the quality of the random space partition for a particular metric space and approximation $c \geq 1$.
Usually, $\rho = 1$ for $c = 1$ and $\rho \to 0$ as $c \to \infty$.

Since the introduction of LSH in \cite{IM}, subsequent research
established optimal values of the LSH exponent $\rho$ for several
metrics of interest, including $\ell_1$ and $\ell_2$. For the Manhattan
distance~($\ell_1$), the optimal value is $\rho=\frac{1}{c} \pm
o(1)$~\cite{IM, MNP, OWZ11}. For the Euclidean metric~($\ell_2$), it is
$\rho=\frac{1}{c^2} \pm o(1)$~\cite{IM, DIIM, AI-CACM, MNP, OWZ11}.

More recently, it has been shown that better bounds on $\rho$ are possible
if random space partitions are \emph{allowed to depend on the
  dataset}\footnote{Let us note that the idea of data-dependent random
  space partitions is ubiquitous in practice, see,
  e.g.,~\cite{wssj-hsss-14,WLKC15-learningHash} for a survey. But the
  perspective in practice is that the given datasets are not ``worst
  case'' and hence it is possible to adapt to the additional ``nice''
  structure.}. That is, the algorithm is based on an observation that every dataset has some structure to exploit. 
This more general framework of \emph{data-dependent LSH}
yields $\rho = \frac{1}{2c - 1} + o(1)$ for the $\ell_1$ distance, and
$\rho = \frac{1}{2c^2 - 1} + o(1)$ for $\ell_2$~\cite{AINR-subLSH,
  r-msthesis-14, AR-optimal}. Moreover, these bounds are known to be
tight for data-dependent LSH~\cite{ar15lower}.

\subsection{Time--space trade-offs}

\label{tradeoff_sec}

Since the early results on LSH, a natural question has been whether
one can obtain query time vs.\ space trade-offs for a fixed approximation $c$. Indeed, data 
structures with \emph{polynomial} space and
\emph{poly-logarithmic} query time were introduced~\cite{IM, KOR00} simultaneously with LSH.

In practice, the most important regime is that of
\emph{near-linear} space, since space is usually a harder constraint than
time: see, e.g., \cite{lv2007multi}. The main
question is whether it is possible to
obtain near-linear space and sublinear query time.
This regime has been studied
since~\cite{I-thesis}, with subsequent improvements in \cite{Pan, AI-CACM,
  lv2007multi, Kap15-tradeoff, AILSR15}. In particular, 
\cite{lv2007multi, AILSR15} introduce practical versions of the above
theoretical results.

Despite significant progress in the near-linear space regime, no known 
algorithms obtain near-linear space and a sublinear query time 
for \emph{all} approximations $c>1$. For
example, the best currently known algorithm of \cite{Kap15-tradeoff}
obtained query time of roughly $n^{4/(c^2+1)}$, which becomes trivial 
for $c<\sqrt{3}$.

\subsection{Lower bounds}
\label{sec:lowerBounds}

Lower bounds for NNS and ANN have also received considerable
attention.  Such lower bounds are ideally obtained in the
\emph{cell-probe} model~\cite{MNSW, miltersen1999cell}, where one measures the \emph{number of memory cells} the
query algorithm accesses. Despite a number of success stories, high
cell-probe lower bounds are notoriously hard to prove. In fact, there
are few techniques for proving high cell-probe lower bounds, for any
(static) data structure problem. For ANN in particular, we have no
viable techniques to prove $\omega(\log n)$ query time lower
bounds. 
Due to this state of affairs, one may rely on
\emph{restricted} models of computation, which nevertheless capture
existing algorithmic approaches.

Early lower bounds for NNS were obtained for data structures in {\em
  exact} or {\em deterministic} settings~\cite{BORl,CCGL, BR,
  Liu-deterNNS, jayram05match, CR04, PT, Y16a}. \cite{CR04, LPY16}
obtained an almost tight cell-probe lower bound for the randomized
Approximate \emph{Nearest} Neighbor Search under the~$\ell_1$
distance. In that problem, there is no distance threshold $r$,
and instead the goal is to find a data point that is not much further
than the \emph{closest} data point.
This
twist is the main source of hardness, so the result
is not applicable to the ANN problem as introduced above.

There are few results that show lower bounds for \emph{randomized}
data structures for ANN.
The first such result~\cite{AIP} shows that any data structure that
solves $(1 + \eps, r)$-ANN for $\ell_1$ or $\ell_2$ using $t$ cell
probes requires space $n^{\Omega(1 / t\eps^2)}$.\footnote{The correct
  dependence on $1/\eps$ requires the stronger Lopsided Set Disjointness lower bound from
  \cite{patrascu11structures}.} This result shows that the algorithms
of~\cite{IM, KOR00} are tight up to constants in the exponent for
$t=O(1)$. 

In~\cite{PTW10} (following up on~\cite{PTW08}), the authors introduce
a general framework for proving lower bounds for ANN under any
metric. They show that lower bounds for ANN are implied by the
\emph{robust expansion} of the underlying metric space.
Using this framework, \cite{PTW10} show that $(c, r)$-ANN using $t$
cell probes requires space $n^{1 + \Omega(1 / tc)}$ for the Manhattan
distance and $n^{1 + \Omega(1 / tc^2)}$ for the Euclidean distance
(for every $c > 1$).

Lower bounds have also been obtained for other metrics.  For the $\ell_\infty$
distance, \cite{ACP08} show a lower bound for deterministic ANN data
structures. This lower bound was later generalized to randomized data
structures \cite{PTW10, KP12-nns}. A~recent result~\cite{AV15} adapts
the framework of~\cite{PTW10} to Bregman divergences.  

To prove higher lower bounds, researchers resorted to lower bounds for
restricted models. These examples include: decision trees
\cite{ACP08} (the corresponding upper bound~\cite{I98} is in the same model),
LSH~\cite{MNP, OWZ11, AILSR15} and data-dependent
LSH~\cite{ar15lower}.

\subsection{Our results: upper bounds}
\label{sec:UBresults}

We give an algorithm obtaining the entire range of time--space
tradeoffs, obtaining sublinear query time for all $c>1$, for the
entire space $\mathbb{R}^d$. Our main theorem is the following:

\begin{theorem}[see Sections~\ref{spherical_sec} and ~\ref{apx:upper_general}]
  \label{main_upper_intro}
  For every $c > 1$, $r > 0$, $\rho_q \geq 0$ and $\rho_u \geq 0$ such that
  \begin{equation}
  \label{main_tradeoff_intro}
  c^2 \sqrt{\rho_q} + \big(c^2 - 1\big) \sqrt{\rho_u} \geq \sqrt{2c^2 - 1},
  \end{equation}
  there exists a data structure for $(c, r)$-ANN for the Euclidean
  space $\Rbb^d$, with space $n^{1 + \rho_u + o(1)} + O(dn)$ and query
  time $n^{\rho_q + o(1)} + dn^{o(1)}$.
\end{theorem}

This algorithm has optimal exponents for all hashing-based algorithms,
as well as one- and two-probe data structures, as we prove in later
sections. In particular, Theorem~\ref{main_upper_intro} recovers or
improves upon all earlier results on ANN in the entire time-space
trade-off. For the near-linear space regime, setting $\rho_u = 0$, we obtain space $n^{1
  + o(1)}$ with query time $n^{\frac{2c^2 - 1}{c^4} + o(1)}$, which is
sublinear for every $c > 1$. For $\rho_q = \rho_u$, we recover the
best data-dependent LSH bound from~\cite{AR-optimal}, with space $n^{1
  + \frac{1}{2c^2 - 1} + o(1)}$ and query time $n^{\frac{1}{2c^2 - 1}
  + o(1)}$. Finally, setting $\rho_q = 0$, we obtain query time
$n^{o(1)}$ and space $n^{\left(\frac{c^2}{c^2 - 1}\right)^2 + o(1)}$,
which, for $c = 1 + \eps$ with $\eps \to 0$, becomes $n^{1 / (4
  \eps^2) + \ldots}$.  

Using a reduction from~\cite{Nguyen-thesis}, we obtain a similar
trade-off for the $\ell_p$ spaces for $1 \leq p < 2$ with $c^2$
replaced with $c^p$.  In particular, for the $\ell_1$ distance we get:
$$
c \sqrt{\rho_q} + \big(c - 1\big) \sqrt{\rho_u} \geq \sqrt{2c - 1}.
$$

Our algorithms can support insertions/deletions with
  only logarithmic loss in space/query time, using the
  \emph{dynamization} technique for decomposable search problems
  from~\cite{ovl81}, achieving update time of $d n^{\rho_u+o(1)}$. To apply this technique,
  one needs to ensure that the preprocessing time is near-linear in the space used,
  which is the case for our data structure.

\subsubsection{Techniques}

We now describe the proof of
Theorem~\ref{main_upper_intro} at a high level.  It consists of two major stages. In
the first stage, we give an algorithm 
for \emph{random} Euclidean instances (introduced formally in
Section~\ref{sec:rand_inst_sec}). In the random Euclidean instances, we generate a dataset
uniformly at random on a unit sphere $S^{d - 1} \subset \Rbb^d$ and
plant a query at random within distance $\sqrt{2} / c$ from a
randomly chosen data point.  In the second stage, we show the claimed
result for the \emph{worst-case} instances by combining ideas from the
first stage with data-dependent LSH from~\cite{AINR-subLSH,
  AR-optimal}.

\paragraph{Data-independent partitions.}
To handle random instances, we use a~certain \emph{data-independent} random process, which we
briefly introduce below.
It can be seen as a modification of spherical Locality-Sensitive Filtering from~\cite{BDGL16-subRho},
and is related to a cell-probe \emph{upper bound} from \cite{PTW10}.
While this data-independent approach can be extended to
\emph{worst case} instances, it gives a bound significantly worse than~\eqref{main_tradeoff_intro}.

We now describe the random process which produces a decision tree to solve 
an instance of ANN on a \emph{Euclidean unit
   sphere} $S^{d - 1} \subset \Rbb^d$. We take our
initial dataset~$P \subset S^{d-1}$ and sample $T$ i.i.d.\ standard
Gaussian $d$-dimensional vectors $z_1$, $z_2$, \ldots, $z_T$. The sets $P_i \subseteq P$ (not necessarily disjoint) are defined for each $z_i$
as follows:
$$
P_i = \{p \in P \mid \langle z_i, p\rangle \geq \Ueta\}.
$$ 
We then recurse and repeat the above procedure for each non-empty
$P_i$.  We stop the recursion once we reach depth $K$. The above procedure
generates a tree of depth $K$ and degree at most $T$, where each leaf explicitly stores
the corresponding subset of the dataset.
To answer a query $q \in S^{d - 1}$, we start at the root
and descend into (potentially multiple) $P_i$'s for which $\langle
z_i, q\rangle \geq \Qeta$.  When we
eventually reach the $K$-th level, we iterate through all the points stored in the accessed leaves searching for a near neighbor.

The parameters $T$, $K$, $\Ueta$ and $\Qeta$ depend on the distance threshold~$r$, the approximation factor~$c$, as well as the desired
space and query time exponents $\rho_u$ and $\rho_q$. The special case of $\Ueta = \Qeta$ corresponds to the ``LSH regime'' $\rho_u = \rho_q$;
$\Ueta < \Qeta$ corresponds to the ``fast queries'' regime $\rho_q < \rho_u$ (the query procedure is more selective);
and $\Ueta > \Qeta$ corresponds to the ``low memory'' regime $\rho_u < \rho_q$. The analysis of this algorithm relies on bounds
on the Gaussian area of certain two-dimensional sets~\cite{AR-optimal, AILSR15}, which are routinely needed for understanding ``Gaussian'' partitions.

This algorithm has two important consequences. First, we obtain the
desired trade-off~\eqref{main_tradeoff_intro} for random instances by
setting $r = \frac{\sqrt{2}}{c}$.  Second, we obtain an inferior
trade-off for \emph{worst-case} instances of $(c, r)$-ANN over a unit
sphere $S^{d - 1}$. Namely, we get:

\begin{equation}
\label{suboptimal_tradeoff_intro}
(c^2 + 1)\sqrt{\rho_q} + (c^2 - 1)\sqrt{\rho_u} \geq 2c.
\end{equation}
Even though it is inferior to the desired bound from~\eqref{main_tradeoff_intro}\footnote{See
  Figure~\ref{trade_off_plot} for comparison for the case $c = 2$.},
it is already non-trivial. In particular,~(\ref{suboptimal_tradeoff_intro}) is better than \emph{all the prior
  work} on time--space trade-offs for ANN, including the most recent trade-off~\cite{Kap15-tradeoff}.  Moreover, using a
reduction from~\cite{Valiant12}, we achieve the
bound~(\ref{suboptimal_tradeoff_intro}) for the whole $\Rbb^d$ as
opposed to just the unit sphere.  Let us formally record it below:
\begin{theorem}
  \label{suboptimal_upper_intro}
  For every $c > 1$, $r > 0$, $\rho_q \geq 0$ and $\rho_u \geq 0$ such
  that~\eqref{suboptimal_tradeoff_intro} holds, there exists a data
  structure for $(c, r)$-ANN for the \emph{whole} $\Rbb^d$ with space
  $n^{1 + \rho_u + o(1)} + O(dn)$ and query time $n^{\rho_q + o(1)} +
  dn^{o(1)}$.
\end{theorem}

\paragraph{Data-dependent partitions.}
We then improve Theorem~\ref{suboptimal_upper_intro} for
worst-case instances and obtain the final result,
Theorem~\ref{main_upper_intro}. We build on the ideas of
data-dependent LSH from~\cite{AINR-subLSH, AR-optimal}.  Using the
reduction from~\cite{Valiant12}, we may assume that the dataset and
queries lie on a unit sphere~$S^{d - 1}$. 

If pairwise distances between data points are
distributed roughly like a random instance, we could apply the data-independent procedure.  In absence of such a
guarantee, we manipulate the dataset in order to reduce it to
a random-looking case. Namely, we look for low-diameter
clusters that contain many data points. We extract these clusters,
and we enclose each of them in a ball of radius non-trivially smaller than one, and we recurse on each cluster. For the remaining points, which do not lie in any cluster, we perform one step of
the data-independent algorithm: we sample $T$ Gaussian vectors,
form $T$ subsets of the dataset, and recurse on each subset. Overall, we make
progress in two ways: for the clusters, we make them a bit more 
isotropic after re-centering, which, after several re-centerings,
makes the instance amenable to the data-independent algorithm, and for
the remainder of the points, we can show that the absence of dense clusters makes
the data-independent algorithm work for a single level of the
tree (though, when recursing into $P_i$'s, dense clusters may re-appear, 
which we will need to extract).

While the above intuition is very simple and, in hindsight, natural,
the actual execution requires a good amount of work. For example, we
need to formalize ``low-diameter'', ``lots of points'', ``more
isotropic'', etc. 
However, compared to~\cite{AR-optimal}, we
manage to simplify certain parts. For example, we do not need to analyze the
behavior of Gaussian partitions on \emph{triples} of points. While this was
necessary in~\cite{AR-optimal}, we can avoid that analysis here, which makes the overall argument much cleaner. The algorithm still requires fine tuning of many moving parts, and we hope that it will be further simplified in the future.

Let us note that prior work suggested that time--space trade-offs might be possible
with data-dependent partitions. To quote~\cite{Kap15-tradeoff}: ``{\em
  It would be very interesting to see if
similar [\ldots to \cite{AINR-subLSH} \ldots] analysis can be used to
improve our tradeoffs}''.

\subsection{Our results: lower bounds}
\label{sec:LBresults}

We show new \emph{cell-probe} and \emph{restricted} lower bounds for $(c,
r)$-ANN matching our upper bounds. All our lower bounds rely on a certain
canonical hard distribution for the Hamming space (defined later in
Section~\ref{sec:rand_inst_sec}). Via a standard reduction~\cite{LLR},
we obtain similar hardness results for $\ell_p$ with $1 < p \leq 2$
(with $c$ being replaced by $c^p$).

\subsubsection{One cell probe}

First, we show a tight lower bound on the
space needed to solve ANN for a random instance, for query algorithms
that use a {\em single} cell probe. More formally, we
prove the following theorem:

\begin{theorem}[see Section~\ref{sec:oneProbe}]
  \label{one_probe_thm}
  Any data structure that:
  \begin{itemize}
        \item solves $(c,r)$-ANN for the Hamming random instance (as defined in
          Section~\ref{sec:rand_inst_sec}) with probability at least~$2/3$,
        \item operates on memory cells of size $n^{o(1)}$,
        \item for each query, looks up a \emph{single} cell,
  \end{itemize}
  must use at least $n^{\left(\frac{c}{c - 1}\right)^2 - o(1)}$ words of memory.
\end{theorem}


The space lower bound matches:
\begin{itemize}
\item Our upper bound for \emph{random instances} that can be made single-probe;
\item Our upper bound for worst-case instances with query time $n^{o(1)}$.
\end{itemize}
The previous best lower bound from~\cite{PTW10} for a single probe
are weaker by a polynomial factor.

We prove Theorem~\ref{one_probe_thm} by computing tight bounds on the
robust expansion of a hypercube $\{-1, 1\}^d$ as defined
in~\cite{PTW10}. Then, we invoke a result from~\cite{PTW10}, which
yields the desired cell probe lower bound. We obtain estimates on the
robust expansion via a combination of the~hypercontractivity
inequality and H\"{o}lder's inequality~\cite{AOBF}. Equivalently, one could obtain the same bounds by an application of the Generalized Small-Set Expansion Theorem 
for $\{-1,1\}^d$ from~\cite{AOBF}.

\subsubsection{Two cell probes}

To state our results for two cell probes, we first define the {\em
  decision} version of ANN (first introduced in \cite{PTW10}).  Suppose
that with every data point $p \in P$ we associate a bit $x_p \in \{0,
1\}$. A new goal is: given a query $q \in \{-1, 1\}^d$ which is within
distance $r$ from a data point $p \in P$, if
$P\setminus \{p\}$ is at distance at least $cr$ from $q$, return $x_p$
with probability at least $2/3$.  It is easy to see that any algorithm
for $(c,r)$-ANN would solve this decision version.

We prove the following lower bound for data structures making
only two cell probes per query.

\begin{theorem}[see Section~\ref{sec:twoProbes}]
  \label{two_probe_thm}
  Any data structure that:
  \begin{itemize}
        \item solves the decision ANN for the random instance
          (Section \ref{sec:rand_inst_sec}) with probability $2/3$,
        \item operates on memory cells of size $o(\log n)$,
        \item accesses at most two cells for each query,
  \end{itemize}
  must use at least $n^{\left(\frac{c}{c - 1}\right)^2 - o(1)}$ words of memory.
\end{theorem}

Informally speaking, Theorem~\ref{two_probe_thm} shows
 that the second cell probe cannot improve
the space bound by more than a subpolynomial factor.  To the best of our
knowledge, this is the first lower bound on the space of \emph{any}
static data structure problem without a polynomial gap between $t=1$
and $t\ge 2$ cell-probes. Previously, the highest ANN lower bound for
two queries was weaker by a polynomial factor~\cite{PTW10}. This
remains the case even if we plug the tight bound on the robust
expansion of a hypercube into the
framework of~\cite{PTW10}. Thus, in order to obtain a higher lower
bound for $t=2$, we must depart from the framework of \cite{PTW10}.

Our proof establishes a connection between two-query data structures
(for the decision version of ANN), and two-query locally-decodable
codes (LDC)~\cite{Y12}. A possibility of such a connection was suggested
in~\cite{PTW08}.  In particular, we show that any data structure
violating the lower bound from Theorem~\ref{two_probe_thm} implies a
too-good-to-be-true
two-query LDC, which contradicts known LDC lower bounds
from~\cite{KW2004, BARW08}.

The first lower bound for unrestricted two-query LDCs was proved
in~\cite{KW2004} via a \emph{quantum} argument. Later, the argument was
simplified and made \emph{classical} in~\cite{BARW08}. It turns out that,
for our lower bound, we need to resort to the original quantum
argument of~\cite{KW2004} since it has a better dependence on the
noise rate a code is able to tolerate.
During the course of our proof, we do not obtain a full-fledged LDC,
but rather an object which can be called an \emph{LDC on average}. For this
reason, we are unable to use~\cite{KW2004} as a black box but rather
adjust their proof to the average case.

Finally, we point out an important difference with
Theorem~\ref{one_probe_thm}: in Theorem~\ref{two_probe_thm} we allow
words to be merely of size $o(\log n)$ (as opposed to
$n^{o(1)}$). Nevertheless, for the \emph{decision version} of ANN for random instances our upper bounds
hold even for such ``tiny''
words.  In fact, our techniques do not allow us to handle words of
size $\Omega(\log n)$ due to the weakness of known lower bounds for
two-query LDC for \emph{large alphabets}.  In particular, our argument
can not be pushed beyond word size $2^{\widetilde{\Theta}(\sqrt{\log
    n})}$ \emph{in principle}, since this would contradict known constructions
of two-query LDCs over large alphabets~\cite{DG15}!

\subsubsection{The general time--space trade-off}

Finally, we prove {\em conditional} lower bound on the entire time--space trade-off matching our upper bounds that up to~$n^{o(1)}$ factors. Note
that---since we show polynomial query time lower bounds---proving
similar lower bounds {\em unconditionally} is far beyond the current
reach of techniques. Any such statement would constitute a major breakthrough in cell probe lower bounds. 

Our lower bounds are proved in the following model, which can be
loosely thought of comprising all hashing-based frameworks we are
aware of:

\begin{definition}
\label{def:lip}
A {\em list-of-points data structure} for the ANN problem is defined as
follows:
\begin{itemize}
\item We fix (possibly random) sets $A_i\subseteq
  \{-1,1\}^d$, for $1 \leq i \leq m$; also, with each possible query point
  $q \in \{-1, 1\}^d$, we associate a (random) set of indices $I(q) \subseteq [m]$;
\item For a given dataset $P$, the data structure maintains $m$ lists of
  points $L_1, L_2, \dots, L_m$, where $L_i=P\cap A_i$;
\item On query $q$, we scan through each list $L_i$ for $i \in I(q)$
  and check whether there exists some $p \in L_i$ with $\|p - q\|_1
  \leq cr$. If it exists, return $p$.
\end{itemize}
The total space is defined as $s = m + \sum_{i=1}^m |L_i|$ and the
query time is $t = |I(q)| + \sum_{i \in I(q)} |L_i|$. 
\end{definition}

For this model, we prove the following theorem.

\begin{theorem}[see Section~\ref{sec:noCoding}]
\label{thm:noCoding}
Consider any list-of-points data structure for $(c, r)$-ANN for random
instances of $n$ points in the $d$-dimensional Hamming space with
$d=\omega(\log n)$, which achieves a total space of $n^{1 + \rho_u}$,
and has query time $n^{\rho_q - o(1)}$, for $2/3$ success
probability. Then it must hold that:
\begin{equation}
\label{eq:power-relation}
c \sqrt{\rho_q} + (c - 1) \sqrt{\rho_u} \geq \sqrt{2c - 1}.
\end{equation}
\end{theorem}

We note that our model captures the basic hashing-based algorithms, in
particular most of the known algorithms for the high-dimensional ANN
problem \cite{KOR00, IM, I98, I-thesis, GIM, Char, DIIM, Pan,
  AC-fastJL, AI-CACM, Pagh16-deterministic, Kap15-tradeoff}, including the
recently proposed Locality-Sensitive Filters scheme from
\cite{BDGL16-subRho}. The only data structures not
captured are the data-dependent schemes from \cite{AINR-subLSH, r-msthesis-14,
  AR-optimal}; we conjecture that the natural extension of the
list-of-point model to data-dependent setting would yield the same
lower bound.  In particular, Theorem \ref{thm:noCoding} uses the
random instance as a hard distribution, for which being
data-dependent seems to offer no advantage.  Indeed, a data-dependent
lower bound in the standard LSH regime (where $\rho_q=\rho_u$) has
been recently shown in \cite{ar15lower}, and matches~\eqref{eq:power-relation} for $\rho_q=\rho_u$.

\subsection{Related work: past and concurrent}

There have been many recent algorithmic advances on high-dimensional
similarity search. The closest pair problem, which can seen as the
off-line version of NNS/ANN, has received much attention
recently~\cite{Valiant12, JW15-closestPair, KKK16-fasterSubquadratic,
  KKKO16, eps13}. ANN solutions with $n^{1+\rho_u}$ space (and
preprocessing), and $n^{\rho_q}$ query time imply closest pair problem
with $O(n^{1+\rho_u}+n^{1+\rho_q})$ time (implying that the balanced,
LSH regime is most relevant). Other work includes locality-sensitive
filters~\cite{BDGL16-subRho} and LSH without false
negatives~\cite{GPY,I00, arasu2006efficient, Pagh16-deterministic,
  PP16}. A step towards bridging the data-depending hashing to the
practical algorithms has been made in \cite{arsn17-lsfForest}.  See
also the surveys \cite{AI-CACM, AI17-handbook}.

\paragraph{Relation to the article of~\cite{christiani16-lower}.}
The article of \cite{christiani16-lower} has significant intersection
with this paper (and, in particular, with the arXiv preprints
\cite{Laarhoven2015,ALRW16} that are now merged to give this paper),
as we explain next. In November 2015, \cite{Laarhoven2015} announced
the optimal trade-off (i.e., Theorem~\ref{main_upper_intro}) for
random instances. As mentioned earlier, it is possible to extend this
result to the entire Euclidean space, albeit with an inferior
trade-off, from Theorem~\ref{suboptimal_upper_intro}; for this, one
can use a standard reduction \'a la \cite{Valiant12} (this extension
was not discussed in~\cite{Laarhoven2015}). On May 9, 2016, both
\cite{christiani16-lower} and \cite{ALRW16} have been announced on
arXiv. In \cite{christiani16-lower}, the author also obtains an upper
bound similar to Theorem \ref{suboptimal_upper_intro} (trade-offs for
the entire $\R^d$, but which are suboptimal), using a different
(data-{\em independent}) reduction from the worst-case to the
spherical case. Besides the upper bound, the author of
\cite{christiani16-lower} also proved a conditional lower bound,
similar to our lower bound from Theorem~\ref{thm:noCoding}. This lower
bound of \cite{christiani16-lower} is independent of our work
in~\cite{ALRW16} (which is now a part of the current paper).

\subsection{Open problems}

We compile a list of exciting open problems:
\begin{itemize}
\item While our upper bounds are optimal (at least, in the hashing framework), the most general algorithms are, unfortunately, impractical. Our trade-offs for random instances on the sphere may well be practical (see also~\cite{BDGL16-subRho, L15a} for an experimental comparison with e.g.~\cite{Char, AILSR15} for $\rho_q = \rho_u$), but a specific bottleneck for the extension to worst-case instances in $\Rbb^d$ is
the clustering step inherited from~\cite{AR-optimal}. Can one obtain
simple and practical algorithms that achieve the optimal time--space
trade-off for these instances as well? For
the balanced regime $\rho_q=\rho_u$, a step in this direction was
taken in \cite{arsn17-lsfForest}.
\item The constructions presented here are optimal when $\omega(\log n) \leq d \leq n^{o(1)}$. Do the same constructions give optimal algorithms in the $d = \Theta(\log n)$ regime? 
\item Our new algorithms for the Euclidean case come tantalizingly close to the best known data structure for the $\ell_{\infty}$ distance~\cite{I98}. 
Can we unify them and extend in a smooth way to the $\ell_p$ spaces for $2 < p < \infty$?
\item Can we improve the dependence on the word size in the reduction from ANN data structures to LDCs used in the two-probe lower bound? As discussed above, the word size can not be pushed beyond $2^{\widetilde{\Theta}(\sqrt{\log n})}$ due to known constructions~\cite{DG15}.
\item A more optimistic view is that LDCs may provide a way to avoid
  the barrier posed by hashing-based approaches. We have shown that
  ANN data structures can be used to build weak forms of LDCs, and an
  intriguing open question is whether known LDC constructions can help
  with designing even more efficient ANN data structures.
\end{itemize}


\section{Random instances}
\label{sec:rand_inst_sec}

In this section, we introduce the \emph{random} instances of ANN for the
Hamming and Euclidean spaces. These instances play a crucial role
for both upper bounds (algorithms) and the lower bounds in all the
subsequent sections (as well as some prior work). 
For upper bounds, we focus on the Euclidean space, since algorithms
for $\ell_2$ yield the algorithms for the Hamming space using standard reductions. 
For the lower bounds, we focus on the Hamming space, since these
yield lower bounds for the Euclidean space.

\paragraph{Hamming distance.} We now describe a distribution supported on dataset-query pairs $(P, q)$, where $P \subset \{-1, 1\}^d$ and $q \in \{-1, 1\}^d$. Random instances of ANN for the Hamming space will be dataset-query pairs drawn from this distribution.
\begin{itemize}
  \item A dataset $P \subset \{-1, 1\}^d$ is given by $n$ points, where each point is drawn independently and uniformly from $\{-1, 1\}^d$, where $d = \omega(\log n)$;
  \item A query $q \in \{-1, 1\}^d$ is drawn by first picking a dataset point $p \in P$ uniformly at random, and then flipping each coordinate of $p$ independently with probability $\frac{1}{2c}$.
  \item The goal of the data structure is to preprocess $P$ in order to recover the data point $p$ from the query point $q$.
\end{itemize}

The distribution defined above is similar to the classic
distribution introduced for the \emph{light bulb problem} in \cite{valiant1988functionality}, 
which can be seen as the \emph{off-line} setting of ANN. This distribution has served as 
the hard distribution in many of the lower bounds for ANN mentioned in Section~\ref{sec:lowerBounds}. 

\paragraph{Euclidean distance.} Now, we describe the distribution supported on dataset-query pairs $(P, q)$, where $P \subset S^{d-1}$ and $q \in S^{d-1}$. Random instances of ANN for Euclidean space will be instances drawn from this distribution.

\begin{itemize}
  \item A dataset $P \subset S^{d-1}$ is given by $n$ unit vectors, where each vector is drawn independently and uniformly at random from $S^{d-1}$. We assume that
  $d = \omega(\log n)$, so pairwise distances are sufficiently concentrated around $\sqrt{2}$.
  \item A query $q \in S^{d-1}$ is drawn by first choosing a dataset
    point $p \in P$ uniformly at random, and then choosing $q$ uniformly at random from all points in $S^{d-1}$ within distance $\frac{\sqrt{2}}{c}$ from $p$.
  \item The goal of the data structure is to preprocess $P$ in order to recover the data point $p$ from the query point $q$.
\end{itemize}

Any data structure for $\left(c + o(1), \frac{\sqrt{2}}{c}\right)$-ANN over $\ell_2$ 
must handle this instance. \cite{AR-optimal} showed how to reduce \emph{any} $(c,r)$-ANN instance to several \emph{pseudo}-random instances without increasing query
time and space too much. These pseudo-random instances have the necessary properties of the random instance above in order for the data-independent algorithms (which are designed with the random instance in mind) to achieve optimal bounds.
Similarly to \cite{AR-optimal}, a data structure for these instances will lie at the core of our algorithm. 

\section{Upper bounds: data-independent partitions}
\label{spherical_sec}

\subsection{Setup}

For $0 < s < 2$, let $\alpha(s) = 1 - \frac{s^2}{2}$ be the cosine of the angle between two points on a unit Euclidean sphere $S^{d - 1}$ with distance $s$ between them,
and $\beta(s) = \sqrt{1 - \alpha^2(s)}$ be the sine of the same angle.

We introduce two functions that will be useful later.
First, for $\eta > 0$, let $$F(\eta) = \Prb{z \sim N(0, 1)^d}{\langle z, u\rangle \geq \eta},$$ where $u \in S^{d - 1}$ is an arbitrary point on the unit
sphere, and $N(0, 1)^d$ is a distribution over $\Rbb^d$, where coordinates of a vector are distributed as i.i.d.\ standard Gaussians. Note that $F(\eta)$ does not depend on the specific choice of $u$ due to the spherical symmetry of Gaussians.

Second, for $0 < s < 2$ and $\eta, \sigma > 0$, let
$$
G(s, \eta, \sigma) = \Prb{z \sim N(0, 1)^d}{\langle z, u\rangle \geq \eta\mbox{ and }\langle z, v\rangle \geq \sigma},
$$
where $u, v \in S^{d - 1}$ are arbitrary points from the unit sphere with $\|u - v\|_2 = s$. As with $F$, the value of $G(s, \eta, \sigma)$ does not
depend on the specific points $u$ and $v$; it only depends on the distance $\|u - v\|_2$ between them. Clearly, $G(s, \eta, \sigma)$ is non-increasing in $s$, for fixed $\eta$ and $\sigma$.

We state two useful bounds on $F(\cdot)$ and $G(\cdot, \cdot, \cdot)$.
The first is a standard tail bound for $N(0, 1)$ and the second follows from a standard computation (see the appendix of \cite{AILSR15} for a proof).
\begin{lemma}
\label{f_bound}
For $\eta \to \infty$,
$$
F(\eta) = e^{-(1 + o(1)) \cdot \frac{\eta^2}{2}}.
$$
\end{lemma}

\begin{lemma}
\label{g_bound}
If $\eta, \sigma \to \infty$, then, for every $s$, one has:
$$
G(s, \eta, \sigma) = e^{-(1 + o(1)) \cdot \frac{\eta^2 + \sigma^2 - 2 \alpha(s) \eta \sigma}{2 \beta^2(s)}}.
$$
\end{lemma}

Finally, by using the Johnson--Lindenstrauss lemma~\cite{JL, DG03} we can assume that $d = \Theta(\log n \cdot \log \log n)$
incurring distortion at most $1 + \frac{1}{\log^{\Omega(1)} \log n}$.

\subsection{Results}

Now we formulate the main result of Section~\ref{spherical_sec}, which we later significantly improve in Section~\ref{apx:upper_general}.

\begin{theorem}
  \label{main_thm_spherical}
  For every $c > 1$, $r > 0$, $\rho_q \geq 0$ and $\rho_u \geq 0$ such that $cr < 2$ and
  \begin{equation}
  \label{alpha_beta_tradeoff}
  \big(1 - \alpha(r) \alpha(cr)\big) \sqrt{\rho_q} + \big(\alpha(r) - \alpha(cr)\big) \sqrt{\rho_u} \geq \beta(r) \beta(cr),
  \end{equation}
  there exists a data structure for $(c, r)$-ANN on a unit sphere $S^{d-1} \subset \Rbb^d$ with space $n^{1 + \rho_u + o(1)}$
  and query time $n^{\rho_q + o(1)}$.
\end{theorem}

We instantiate Theorem~\ref{main_thm_spherical} for two important cases. First, we get a single trade-off between $\rho_q$
and $\rho_u$ \emph{for all $r > 0$ at the same time} by observing that~(\ref{alpha_beta_tradeoff})
is the worst when $r \to 0$. Thus, we get a bound on $\rho_q$ and $\rho_u$ that depends on
the approximation $c$ only, which then can easily be translated to a result for the \emph{whole}~$\Rbb^d$ using a reduction
from~\cite{Valiant12}.

\begin{corollary}
  \label{worst_case_corollary}
  For every $c > 1$, $r > 0$, $\rho_q \geq 0$ and $\rho_u \geq 0$ such that
  \begin{equation}
  \label{worst_case_tradeoff}
  \big(c^2 + 1\big) \sqrt{\rho_q} + \big(c^2 - 1\big) \sqrt{\rho_u} \geq 2c,
  \end{equation}
  there exists a data structure for $(c, r)$-ANN for the \emph{whole} $\Rbb^d$ with space $n^{1 + \rho_u + o(1)}$
  and query time $n^{\rho_q + o(1)}$.
\end{corollary}

\begin{proof}
We will show that we may transform an instance of $(c, r)$-ANN on $\Rbb^d$ to an instance of $(c+o(1), r')$-ANN on the sphere with $r'\to 0$. 
When $r' \to 0$, we have:
\begin{align*}
1 - \alpha(r') \alpha(cr') &= \frac{(c^2 + 1)r'^2}{2} + O_c(r'^4),\\
\alpha(r') - \alpha(cr') &= \frac{(c^2 - 1)r'^2}{2} + O_c(r'^4),\\
\beta(r') \beta(cr') &= c r'^2 + O_c(r'^4).
\end{align*}
Substituting these estimates into~(\ref{alpha_beta_tradeoff}), we get~(\ref{worst_case_tradeoff}).

Now let us show how to reduce ANN over $\Rbb^d$ to the case, when all the points and queries lie on a unit sphere.

We first rescale all coordinates so as to assume $r = 1$. Now let us partition the whole space $\Rbb^d$ into randomly shifted cubes with the side length $s = 10 \cdot \sqrt{d}$ and consider each cube separately. For any query $q \in \Rbb^d$, with near neighbor $p \in P$,
\[ \Pr[\text{$p$ and $q$ are in different cubes}] \leq \sum_{i=1}^d \dfrac{|p_i - q_i|}{s} = \dfrac{\|p-q\|_1}{s} \leq \dfrac{\sqrt{d} \cdot \|p-q\|_2}{s} \leq \dfrac{1}{10}. \]
The $\ell_2$ diameter of a single cube is $d$. Consider one particular cube $C$, where we first translate points so $x \in C$ have $\|x\|_2 \leq d$. We let $\pi \colon C \to \R^{d+1}$ where
\[ \pi(x) = (x, R), \]
where we append coordinate $R \gg d$ as the $(d+1)$-th coordinate.
For any point $x \in C$, 
\[ \left\| \pi(x) - \left(\frac{R}{\|\pi(x)\|_2}\right)\cdot \pi(x) \right\|_2 \leq \dfrac{\|x\|_2^2}{2R}\]
and for any two points $x, y \in C$, $\|x - y\|_2 = \|\pi(x) - \pi(y)\|_2$; thus,
$$\left\|\left( \frac{R}{\|\pi(x)\|_2}\right)\pi(x) - \left(\frac{R}{\|\pi(y)\|_2}\right) \pi(y)\right\|_2 \leq \frac{d^2}{R} + \|x-y\|_2.
$$
In addition, since $\left( \frac{R}{\|\pi(x)\|_2}\right) \pi(x)$ lies in a sphere of radius $R$ for each point $x \in C$. Thus, letting $R = d^2 \cdot \log \log n \leq O(\log^2 n \cdot \log^3 \log n)$ (which is without loss of generality by the Johnson--Lindenstrauss Lemma), we get that an instance of $(c, r)$-ANN on $\R^d$ corresponds to an instance of $(c+o(1), \frac{1}{d^2 \log \log n})$-ANN on the surface of the unit sphere $S^d \subset \R^{d+1}$, where we lose $\frac{1}{10}$ in the success probability due to the division into disjoint cubes. Applying Theorem~\ref{main_thm_spherical}, we obtain the desired bound.
\end{proof}

If we instantiate Theorem~\ref{main_thm_spherical} with inputs (dataset and query) drawn from the \emph{random} instances defined in Section~\ref{sec:rand_inst_sec} (corresponding to the case $r = \frac{\sqrt{2}}{c}$), we obtain a significantly better tradeoff than~(\ref{worst_case_tradeoff}). By simply applying Theorem~\ref{main_thm_spherical}, we give a trade-off for random instances matching the trade-off promised in Theorem~\ref{main_upper_intro}.

\begin{corollary}
  \label{random_corollary}
  For every $c > 1$, $\rho_q \geq 0$ and $\rho_u \geq 0$ such that
  \begin{equation}
  \label{random_tradeoff}
  c^2 \sqrt{\rho_q} + \big(c^2 - 1\big) \sqrt{\rho_u} \geq \sqrt{2c^2 - 1},
  \end{equation}
  there exists a data structure for $\left(c, \frac{\sqrt{2}}{c}\right)$-ANN on a unit sphere $S^{d-1} \subset \Rbb^d$ with space $n^{1 + \rho_u + o(1)}$
  and query time $n^{\rho_q + o(1)}$. In particular, this data structure is able to handle random instances as defined in Section~\ref{sec:rand_inst_sec}.
\end{corollary}
\begin{proof}
Follows from (\ref{alpha_beta_tradeoff}) and that $\alpha(\sqrt{2}) = 0$ and $\beta(\sqrt{2}) = 1$.
\end{proof}

Figure~\ref{trade_off_plot} plots the time-space trade-off in (\ref{worst_case_tradeoff}) and (\ref{random_tradeoff}) for $c = 2$. Note that~(\ref{random_tradeoff}) is much better than~(\ref{worst_case_tradeoff}), especially when $\rho_q = 0$, where
(\ref{worst_case_tradeoff}) gives space $n^{2.77\ldots}$, while (\ref{random_tradeoff}) gives much better space $n^{1.77\ldots}$.
In Section~\ref{apx:upper_general}, we show how to get best of both worlds: we obtain the trade-off~(\ref{random_tradeoff}) for \emph{worst-case} instances.
The remainder of the section is devoted to proving Theorem~\ref{main_thm_spherical}.

\subsection{Data structure}

\subsubsection{Description}

Fix $K$ and $T$ to be positive integers, we determine their exact value later. Our data structure is a \emph{single} rooted tree where each node corresponds to a spherical cap. The tree consists of $K + 1$ levels of nodes where each node has out-degree at most $T$. We will index the levels by $0$, $1$, \ldots, $K$, where the $0$-th level consists of the root denoted by $v_0$, and each node up to the $(K-1)$-th level has at most $T$ children. Therefore, there are at most $T^K$ nodes at the $K$-th level. 

For every node $v$ in the tree, let~$\Lc_v$ be the set of nodes on the path from $v$ to the root $v_0$ excluding the root (but including $v$).
Each node $v$, except for the root, stores a random Gaussian vector $z_v \sim N(0, 1)^d$. For each node~$v$, we define the following subset of the dataset
$P_v \subseteq P$:
$$
P_v = \left\{p \in P \mid \forall v' \in \Lc_v \enspace \langle z_{v'}, p \rangle \geq \Ueta\right\},
$$
where $\Ueta > 0$ is a parameter to be chosen later. 

At the root node $v_0$, $P_{v_0} = P$, since $\Lc_{v_0} = \emptyset$. Intuitively, each set $P_v$ corresponds to a subset of the dataset lying in the intersection of spherical caps centered around $z_{v'}$ for all $v' \in \Lc_v$. Every leaf $\ell$ at the level~$K$ stores the subset $P_\ell$ explicitly. 

We build the tree recursively. For a given node $v$ in levels $0$, \ldots, $K-1$,
we first sample $T$ i.i.d.\ Gaussian vectors $g_1, g_2, \ldots, g_T \sim N(0, 1)^d$. Then, for every $i$ such that $\{p \in P_v\mid \langle g_i, p\rangle \geq \Ueta\}$
is non-empty, we create a new child $v'$ with $z_{v'} = g_i$ and recursively process~$v'$. At the $K$-th level, each node $v$ stores $P_v$ as a list of points.

In order to process a query $q \in S^{d-1}$, we start from the root $v_0$ and descend down the tree. We consider every child $v$ of the root for which $\langle z_v, q\rangle \geq \Qeta$,
where $\Qeta > 0$ is another parameter to be chosen later\footnote{Note that $\Ueta$ may not be equal to $\Qeta$. It is exactly this discrepancy that will govern the time--space trade-off.}. After identifying all the children, we proceed down the children recursively. If we reach leaf $\ell$ at level $K$, we scan through all the points in $P_\ell$ and compute their distance to the query~$q$. If a point lies at a distance at most $cr$ from the query, we return it and stop.

We provide pseudocode for the data structure above in Figure~\ref{spherical_pseudo}. The procedure \textsc{Build}($P$, $0$, $\perp$) builds the data structure for dataset $P$ and returns the root of the tree, $v_0$. The procedure \textsc{Query}($q$, $v_0$) queries the data structure with root $v_0$ at point $q$. 

\begin{figure*}[t!]
\begin{subfigure}[t]{0.5\textwidth}
\centering
{\footnotesize
\begin{algorithmic}
\Function{Build}{$P'$, $l$, $z$}
\State create a tree node $v$
\State store $l$ as $v.l$
\State store $z$ as $v.z$
\If{$l = K$}
\State store $P'$ as $v.P$
\Else
\For{$i \gets 1 \ldots T$}
\State sample a Gaussian vector $z' \sim N(0, 1)^d$
\State $P'' \gets \{p \in P' \mid \langle z', p \rangle \geq \Ueta\}$
\If{$P'' \ne \emptyset$}
\State add \Call{Build}{$P''$, $l + 1$, $z'$} as a child of $v$
\EndIf
\EndFor
\EndIf
\State \Return{$v$}
\EndFunction
\end{algorithmic} }
\end{subfigure}
~\qquad \qquad
\begin{subfigure}[t]{0.5\textwidth}
\centering
{\footnotesize
\begin{algorithmic}
\Function{Query}{$q$, $v$}
\If{$v.l = K$}
\For{$p \in v.P$}
\If{$\|p - q\| \leq cr$}
\State \Return $p$
\EndIf
\EndFor
\Else
\For{$v'$ : $v'$ is a child of $v$}
\If{$\langle v'.z, q\rangle \geq \Qeta$}
\State $p \gets \Call{Query}{q, v'}$
\If{$p \ne \perp$}
\State \Return $p$
\EndIf
\EndIf
\EndFor
\EndIf
\State \Return $\perp$
\EndFunction
\end{algorithmic}
}
\end{subfigure}
\caption{Pseudocode for data-independent partitions}
\label{spherical_pseudo}
\end{figure*}

\subsubsection{Analysis}
\label{analysis_sec}

\paragraph{Probability of success} We first analyze the probability of success of the data structure. We assume that a query $q$ has some $p \in P$ where $\|p - q\|_2 \leq r$. The data structure succeeds when \textsc{Query}($q$, $v_0$) returns some point $p' \in P$ with $\|q - p'\|_2 \leq cr$.

\begin{lemma}
\label{lem:succ-prob}
  If $$T \geq \frac{100}{G\left(r, \Ueta, \Qeta\right)},$$ then
   with probability at least $0.9$, \textsc{Query}\emph{(}$q$, $v_0$\emph{)} finds some point within distance $cr$ from $q$.
\end{lemma}
\begin{proof}
We prove the lemma by induction on the depth of the tree. Let $q \in S^{d - 1}$ be a query point and $p \in P$ its near neighbor. Suppose we are within the recursive call \textsc{Query}($q$, $v$) for some node $v$ in the tree. Suppose we have not yet failed, that is, $p \in P_v$. We would like to prove that---if the condition of
the lemma is met---the probability that this call returns \emph{some} point within distance $cr$ is at least~$0.9$.
  
When $v$ is a node in the last level $K$, the algorithm enumerates $P_v$ and, since we assume $p \in P_v$, some
good point will be discovered (though not necessarily $p$ itself). Therefore, this case is trivial. 
Now suppose that $v$ is not from the $K$-th level. Using the inductive assumption,
suppose that the statement of the lemma is true for all $T$ potential children of~$v$, i.e., if $p \in P_{v'}$, then with probability 0.9, \textsc{Query}$(q, v')$ returns some point within distance $cr$ from $q$. Then,
  \begin{align*}
\Pr[\mbox{failure}] &\leq \prod_{i=1}^T \left(1 - \Pr_{z_{v_i}}\left[\langle z_{v_i}, p \rangle \geq \Ueta\mbox{ and }\langle z_{v_i}, q \rangle \geq \Qeta\right] \cdot 0.9\right) \\ 
                          &\leq \left(1 - G\left(r, \Ueta, \Qeta\right) \cdot 0.9\right)^T \leq 0.1,
  \end{align*}
  where the first step follows from the inductive assumption and independence between the children of $v$ during the preprocessing phase.
  The second step follows by monotonicity of $G(s, \rho, \sigma)$ in $s$, and the third step is due to the assumption of the lemma.
\end{proof}

\paragraph{Space} We now analyze the space consumption of the data structure.
\begin{lemma}
\label{lem:space}
  The expected space consumption of the data structure is at most
  $$
  n^{1 + o(1)} \cdot K \cdot \big(T \cdot F(\Ueta)\big)^K.
  $$
\end{lemma}
\begin{proof}
  We compute the expected total size of the sets $P_\ell$ for leaves $\ell$ at $K$-th level. There are at most $T^K$ such nodes,
  and for a fixed point $p \in P$ and a fixed leaf $\ell$ the probability that $p \in P_\ell$ is equal to $F(\Ueta)^K$. Thus, the expected total
  size is at most $n \cdot \big(T \cdot F(\Ueta)\big)^K$. Since we only store a node $v$ if $P_v$ is non-empty, the number of nodes stored is at most $K+1$ times the number of points stored at the leaves. The Gaussian vectors stored at each node require space $d$, which is at most $n^{o(1)}$.
\end{proof}

\paragraph{Query time} Finally, we analyze the query time.
\begin{lemma}
\label{lem:time}
  The expected query time is at most
  \begin{equation}
  \label{query_time_full_bound}
  n^{o(1)} \cdot T \cdot \left( T \cdot F(\Qeta)\right)^K
  + n^{1 + o(1)} \cdot \left(T \cdot G(cr, \Ueta, \Qeta)\right)^K.
  \end{equation}
\end{lemma}
\begin{proof}
  First, we compute the expected query time spent going down the tree, without scanning the leaves. The expected number of \emph{nodes} the query procedure reaches is:
  $$
  1 + T \cdot F(\Qeta) + \left(T \cdot F(\Qeta)\right)^2 + \ldots + \left(T \cdot F(\Qeta)\right)^K = O(1) \cdot \left(T \cdot F(\Qeta)\right)^K,
  $$
  since we will set $T$ so $T \cdot F(\eta_q) \geq 100$. 
In each of node, we spend time $n^{o(1)} \cdot T$. The product of the two expressions gives the first term in (\ref{query_time_full_bound}).
  
The expected time spent scanning points in the leaves is at most $n^{o(1)}$ times the number of points scanned at the leaves reached. The number of points scanned is always at most one more than the number of \emph{far} points, i.e., lying a distance greater than $cr$ from $q$, that reached the same leaf. There are at most $n-1$ far points and $T^K$ leaves.
  For each far point $p'$ and each leaf $\ell$ the probability that both $p'$ and $q$ end up in $P_\ell$ is at most
  $G(cr, \Ueta, \Qeta)^K$. For each such pair, we spend time at most $n^{o(1)}$
  processing the corresponding $p'$. This gives the second term in (\ref{query_time_full_bound}).
\end{proof}

\subsubsection{Setting parameters}
\label{par_set}

We end the section by describing how to set parameters $T$, $K$, $\Ueta$ and $\Qeta$ to prove Theorem~\ref{main_thm_spherical}.

First, we set $K \sim \sqrt{\ln n}$. In order to satisfy the requirement
of Lemma~\ref{lem:succ-prob}, we set
\begin{equation}
\label{t_vs_g}
T = \frac{100}{G(r, \Ueta, \Qeta)}.
\end{equation}
Second, we (approximately) balance the terms in the query time~(\ref{query_time_full_bound}).
Toward this goal, we aim to have
\begin{equation}
\label{f_vs_g}
F(\Qeta)^K = n \cdot G(cr, \Ueta, \Qeta)^K.
\end{equation} 
If we manage to satisfy these conditions, then we obtain space $n^{1 + o(1)} \cdot \left(T \cdot F(\Ueta)\right)^K$ and query time\footnote{Other terms from the query time are absorbed into $n^{o(1)}$ due to our choice of $K$.}
$n^{o(1)} \cdot \left(T \cdot F(\Qeta)\right)^K$.

Let $F(\Ueta)^K = n^{-\sigma}$ and $F(\Qeta)^K = n^{-\tau}$. By Lemma~\ref{f_bound}, Lemma~\ref{g_bound} and~(\ref{f_vs_g}), we have that, up to $o(1)$ terms,
$$
\tau = \frac{\sigma + \tau - 2 \alpha(cr) \cdot \sqrt{\sigma\tau}}{\beta^2(cr)} - 1,
$$
which can be rewritten as
\begin{equation}
\label{balancing}
\big|\sqrt{\sigma} - \alpha(cr) \sqrt{\tau}\big| = \beta(cr),
\end{equation}
since $\alpha^2(cr) + \beta^2(cr) = 1$.
We have, by Lemma~\ref{f_bound}, Lemma~\ref{g_bound} and~(\ref{t_vs_g}),
$$
T^K = n^{\frac{\sigma + \tau - 2 \alpha(r) \sqrt{\sigma \tau}}{\beta^2(r)} + o(1)}.
$$
Thus, the space bound is
$$
n^{1 + o(1)} \cdot \left(T \cdot F(\Ueta)\right)^K = n^{1 + \frac{\sigma + \tau - 2 \alpha(r) \sqrt{\sigma \tau}}{\beta^2(r)} - \sigma + o(1)}
= n^{1 + \frac{\left(\alpha(r) \sqrt{\sigma} - \sqrt{\tau}\right)^2}{\beta^2(r)} + o(1)}
$$
and query time is
$$
n^{o(1)} \cdot \left(T \cdot F(\Qeta)\right)^K = n^{\frac{\sigma + \tau - 2 \alpha(r) \sqrt{\sigma \tau}}{\beta^2(r)} - \tau + o(1)}
= n^{\frac{\left(\sqrt{\sigma} - \alpha(r)\sqrt{\tau}\right)^2}{\beta^2(r)} + o(1)}.
$$
In other words,
$$
\rho_q = \frac{\big(\sqrt{\sigma} - \alpha(r) \sqrt{\tau}\big)^2}{\beta^2(r)},
$$
and
$$
\rho_u = \frac{\big(\alpha(r)\sqrt{\sigma} - \sqrt{\tau}\big)^2}{\beta^2(r)}
$$
where $\tau$ is set so that~(\ref{balancing}) is satisfied. Combining these identities, we obtain~(\ref{alpha_beta_tradeoff}).

Namely, we set $\sqrt{\sigma} = \alpha(cr) \sqrt{\tau} + \beta(cr)$ to satisfy~(\ref{balancing}).
Then, $\sqrt{\tau}$ can vary between:
$$
\frac{\alpha(r)\beta(cr)}{1 - \alpha(r) \alpha(cr)},
$$
which corresponds to $\rho_u = 0$ and
$$
\frac{\beta(cr)}{\alpha(r) - \alpha(cr)},
$$
which corresponds to $\rho_q = 0$.

This gives a relation:
$$
\sqrt{\tau} = \frac{\beta(cr) - \beta(r) \sqrt{\rho_q}}{\alpha(r) - \alpha(cr)} = \frac{\alpha(r)\beta(cr) + \beta(r) \sqrt{\rho_u}}{1 - \alpha(r) \alpha(cr)},
$$
which gives the desired trade-off~(\ref{alpha_beta_tradeoff}).

\subsubsection{An algorithm based on Locality-Sensitive Filtering (LSF)}

We remark that there is an alternative method to the algorithm described
above, using {\em Spherical Locality-Sensitive
  Filtering} introduced in~\cite{BDGL16-subRho}.
  As argued in \cite{BDGL16-subRho}, this method may naturally
extend to the $d=O(\log n)$ case with better trade-offs
between $\rho_q, \rho_u$ than in~\eqref{main_tradeoff_intro}
(indeed, such better exponents were obtained in~\cite{BDGL16-subRho}
for the ``LSH regime'' of $\rho_u=\rho_q$).

For spherical LSF, partitions are formed by first dividing $\Rbb^d$ into $K$
blocks ($\Rbb^d = \Rbb^{d/K} \times \dots \times \Rbb^{d/K}$), and
then generating a spherical code $C \subset S^{d/K - 1} \subset
\Rbb^{d/K}$ of vectors sampled uniformly at random from the
lower-dimensional unit sphere $S^{d/K - 1}$. For any vector $p \in
\Rbb^d$, we write $p^{(1)}, \dots, p^{(K)}$ for the $K$ blocks of
$d/K$ coordinates in the vector $p$. 

The tree consists of $K$ levels, and the $|C|$ children of a node $v$ at level $\ell$ are defined by the vectors $(0, \dots, 0, z_i, 0, \dots, 0)$, where only the $\ell$-th block of $d/K$ entries is potentially non-zero and is formed by one of the $|C|$ code words. The subset $P''$ of a child then corresponds to the subset $P'$ of the parent, intersected with the spherical cap corresponding to the child, where
\begin{align}
P' = \{p \in P: \langle z_{i_1}, p^{(1)} \rangle + \dots + \langle z_{i_K}, p^{(K)} \rangle \geq K \cdot \Ueta\}.
\end{align}
Decoding each of the $K$ blocks separately with threshold $\Ueta$ was shown in \cite{BDGL16-subRho} to be asymptotically equivalent to decoding the entire vector with threshold $K \cdot \Ueta$, as long as $K$ does not grow too fast as a function of $d$ and $n$. The latter joint decoding method based on the sum of the partial inner products is then used as the actual decoding method.

\begin{figure}
\centering
\includegraphics[scale=0.5]{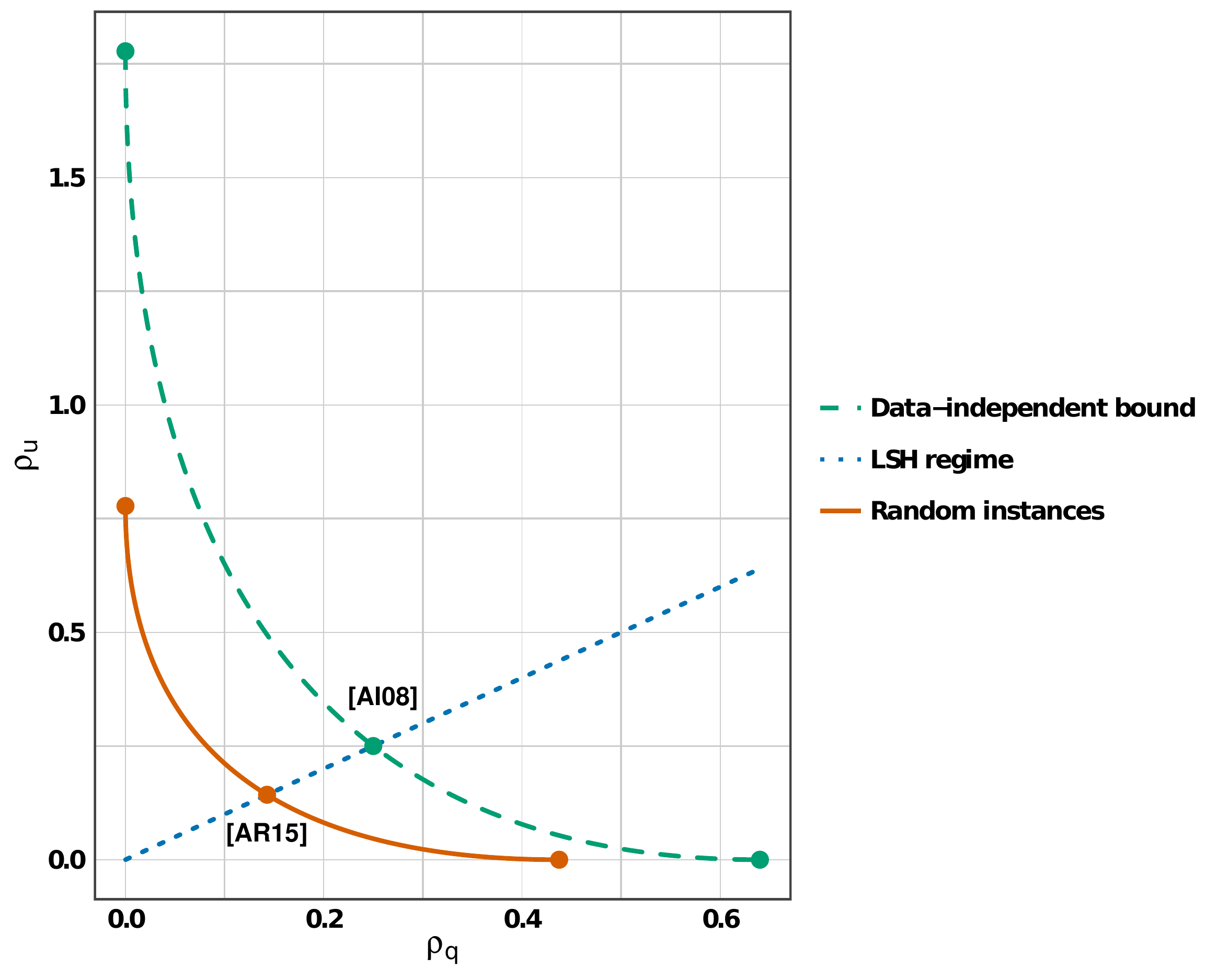}
\caption{Trade-offs between query time $n^{\rho_q + o(1)}$ and space $n^{1 + \rho_u + o(1)}$ for the Euclidean distance
and approximation $c = 2$. The green dashed line corresponds to
the simple data-independent bound for \emph{worst-case} instances from Corollary~\ref{worst_case_corollary}. The red solid line corresponds to
the bound for \emph{random} instances from Corollary~\ref{random_corollary}, which we later extend
to \emph{worst-case} instances in Section~\ref{apx:upper_general}. The blue dotted line is $\rho_q = \rho_u$, which corresponds to the ``LSH regime''. In particular, the intersection
of the dotted and the dashed lines matches the best \emph{data-independent} LSH from~\cite{AI-CACM}, while the intersection with the solid line
matches the best \emph{data-dependent} LSH from~\cite{AR-optimal}.}
\label{trade_off_plot}
\end{figure}


\section{Upper bounds: data-dependent partitions}
\label{apx:upper_general}

In this section we prove the main upper bound theorem,
Theorem~\ref{main_upper_intro}, which we restate below:

\begin{theorem}
  \label{worst_case_data_dependent}
  For every $c > 1$, $r > 0$, $\rho_q \geq 0$ and $\rho_u \geq 0$ such that
  \begin{equation}
  \label{worst_case_tradeoff_again}
  c^2 \sqrt{\rho_q} + \big(c^2 - 1\big) \sqrt{\rho_u} \geq \sqrt{2c^2 - 1},
  \end{equation}
  there exists a data structure for $(c, r)$-ANN for the \emph{whole} $\Rbb^d$ with space $n^{1 + \rho_u + o(1)} + O(dn)$
  and query time $n^{\rho_q + o(1)} + dn^{o(1)}$.
\end{theorem}

This theorem achieves ``the best of both worlds'' in Corollary~\ref{worst_case_corollary} and
Corollary~\ref{random_corollary}. Like Corollary~\ref{worst_case_corollary}, our data structure works for worst-case datasets; however, we improve upon the trade-off between time and space complexity from Corollary~\ref{worst_case_corollary} to that of random instances in Corollary~\ref{random_corollary}. 
See Figure~\ref{trade_off_plot} for a comparison of both trade-offs for $c = 2$. We achieve the improvement by combining the result of Section~\ref{spherical_sec} with the techniques from~\cite{AR-optimal}.

As in~\cite{AR-optimal}, the resulting data structure is a decision
tree. However, there are several notable differences from
\cite{AR-optimal}:
\begin{itemize}
  \item The whole data structure is a \emph{single} decision tree, while~\cite{AR-optimal} considers a \emph{collection} of $n^{\Theta(1)}$ trees. 
  \item Instead of Spherical LSH used in~\cite{AR-optimal}, we use the partitioning procedure from Section~\ref{spherical_sec}.
  \item In~\cite{AR-optimal}, the algorithm continues partitioning the dataset until all parts contain less than $n^{o(1)}$ points. We change the stopping criterion slightly to ensure the number of ``non-cluster'' nodes on any root-leaf branch is the same (this value will be around $\sqrt{\ln n}$ to reflect the setting of $K$ in Section~\ref{spherical_sec}).
  \item Unlike~\cite{AR-optimal}, our analysis does not require the ``three-point property'', which is necessary in~\cite{AR-optimal}. This is related to the fact
  that the probability success of a single tree is constant, unlike~\cite{AR-optimal}, where it is polynomially small.
  \item In~\cite{AR-optimal}, the algorithm reduces the general case to the ``bounded ball'' case using LSH from~\cite{DIIM}. While the cost associated with this procedure is negligible in the LSH regime, the cost becomes too high in certain parts of the time--space trade-off. Instead, we use a standard trick of imposing a randomly shifted grid, which reduces an arbitrary dataset to a dataset of diameter
  $\widetilde{O}(\log n)$ (see the proof of Corollary~\ref{worst_case_tradeoff} and \cite{IM}). Then, we invoke an upper bound from Section~\ref{spherical_sec} together with a reduction from~\cite{Valiant12} which happens to be enough for this case. 
  \end{itemize}

\subsection{Overview}
\label{sec_overview}

We start with a high-level overview. Consider a dataset $P_0$ of $n$ points. We may assume $r = 1$ by rescaling. 
We may further assume the
dataset lies in the Euclidean space of dimension $d = \Theta(\log n
\cdot \log \log n)$; one can always reduce the dimension to $d$ by
applying the Johnson--Lindenstrauss lemma~\cite{JL, DG03} which reduces the dimension and distorts pairwise distances by at most $1 \pm 1 / (\log \log n)^{\Omega(1)}$ with high probability. We may also assume the entire dataset $P_0$ and a query lie on a sphere $\partial B(0, R)$ of radius $R = \widetilde{O}(\log^2 n)$ (see the proof of Corollary~\ref{worst_case_tradeoff}).

We partition $P_0$ into various components: $s$ {\em dense} components, denoted by $C_1$, $C_2$, \ldots, $C_s$, and one {\em pseudo-random} component, denoted by $\widetilde{P}$. The partition is designed to satisfy the following properties.
Each dense component $C_i$ satisfies $|C_i| \geq \tau n$ and
 can be covered by a spherical cap of radius $(\sqrt{2} -
\eps) R$ (see Figure~\ref{cap_covering_fig}). Here $\tau, \eps > 0$ are small quantities to be chosen
later. One should think of $C_i$ as clusters consisting of $n^{1 - o(1)}$ points which are closer than random points would be. The pseudo-random component $\widetilde{P}$ consists of the remaining points without any dense clusters inside.

\begin{figure}
    \begin{center}
        \includegraphics[page=2,scale=0.8]{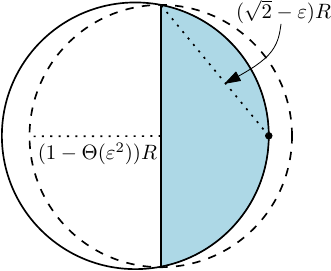}
    \end{center}
    \caption{Covering a spherical cap of radius $(\sqrt{2} - \eps) R$}
    \label{cap_covering_fig}
\end{figure}

We proceed separately for each $C_i$ and $\widetilde{P}$.
We enclose every dense component $C_i$ in a smaller ball $E_i$ of
radius $(1 - \Omega(\eps^2)) R$ (see Figure~\ref{cap_covering_fig}).
For simplicity, we first ignore the fact that $C_i$ does not
necessarily lie on the boundary $\partial E_i$.  Once we enclose each dense cluster with a smaller ball, we recurse on each resulting spherical instance of radius $(1 -
\Omega(\eps^2)) R$.  We treat the pseudo-random component $\widetilde{P}$
similarly to the random instance from Section~\ref{sec:rand_inst_sec} described in Section~\ref{spherical_sec}. Namely, we sample $T$ Gaussian vectors
$z_1, z_2, \ldots, z_T \sim N(0, 1)^d$, and form $T$ subsets of~$\widetilde{P}$:
$$
\widetilde{P}_i = \{p \in \widetilde{P} \mid \langle z_i , p \rangle \geq \Ueta R\},
$$ where $\Ueta > 0$ is a parameter to be chosen later (for each
pseudo-random remainder separately). Then, we recurse on each
$\widetilde{P}_i$. Note that after we recurse, new dense clusters may appear in some $\widetilde{P}_i$ since it becomes easier to satisfy the minimum size constraint.

During the query procedure, we recursively
query \emph{each} $C_i$ with the query point $q$. For the pseudo-random component
$\widetilde{P}$, we identify all $i$'s such that $\langle z_i, q
\rangle \geq \Qeta R$, and query all corresponding children
recursively. Here $T$, $\Ueta > 0$ and $\Qeta > 0$ are parameters
that need to be chosen carefully (for
each pseudo-random remainder separately).

Our algorithm makes progress in two ways.
For dense clusters, we reduce the radius of the enclosing sphere by a factor
of $(1 - \Omega(\eps^2))$. Ideally, we have that initially $R = \widetilde{O}(\log^2 n)$, so in $O(\log \log n / \eps^2)$ iterations of removing dense clusters, we
arrive at the case of $R \leq c /\sqrt{2}$, where Corollary~\ref{random_corollary} begins to apply. For the pseudo-random component $\widetilde{P}$, most points will lie a distance at least $(\sqrt{2} - \eps) R$ from each other. In particular, the ratio of $R$ to a typical
inter-point distance is approximately $1/\sqrt{2}$, exactly like in a random case. For this reasonm we call $\widetilde{P}$ pseudo-random.
In this setting, the data structure from Section~\ref{spherical_sec} performs well.

We now address the issue deferred in the above high-level
description: that a dense component~$C_i$ does not generally
lie on~$\partial E_i$, but rather can occupy the interior of $E_i$. In this case, we partition~$E_i$ into very thin annuli of carefully
chosen width and treat each annulus as a sphere. This
discretization of a ball adds to the complexity of the analysis,
but is not fundamental from the conceptual point of
view.

\subsection{Description}

We are now ready to describe the data structure formally. It depends
on the (small positive) parameters $\tau$, $\eps$ and $\delta$, as well as an integer parameter $K \sim \sqrt{\ln n}$.
We also need to choose parameters $T$, $\Ueta > 0$, $\Qeta > 0$ for each pseudo-random remainder separately. Figure~\ref{pseudocode} provides pseudocode for the algorithm.

\begin{figure}
\centering
    \begin{center}
        \includegraphics[width=0.4\linewidth]{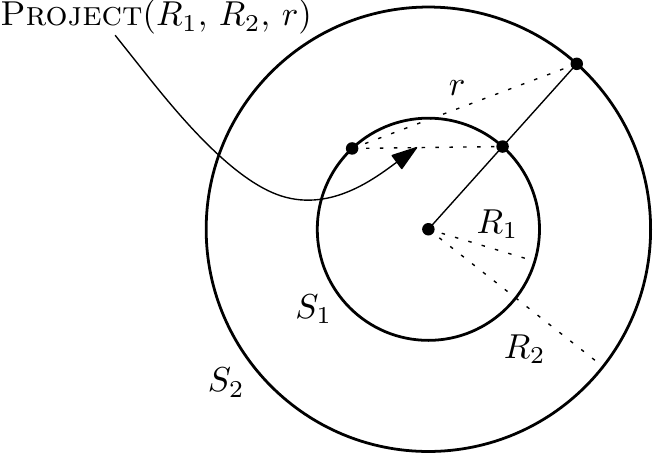}
    \end{center}
    \caption{The definition of \textsc{Project}}
    \label{project_fig}
\end{figure}

\paragraph{Preprocessing.}
Our preprocessing algorithm consists of the following functions:
\begin{itemize}
   \item \textsc{Process}($P$) does the initial preprocessing. In particular, it performs the rescaling so that $r_1 = 1$ as well as the dimension reduction to $d = \Theta(\log n \log \log n)$ with the Johnson--Lindenstrauss lemma~\cite{JL, DG03}. In addition, we partition into randomly shifted cubes, translate the points, and think of them as lying on a sphere of radius $R = \widetilde{O}(\log^2 n)$ (see the proof of Corollary~\ref{worst_case_tradeoff} for details). Then we call \textsc{ProcessSphere}.
    \item \textsc{ProcessSphere}($P$, $r_1$, $r_2$, $o$, $R$, $l$) builds the data structure for a dataset $P$ lying on a sphere $\partial B(o, R)$,
    assuming we need to solve ANN with distance thresholds $r_1$ and $r_2$. Moreover, we are guaranteed that queries
    will lie on $\partial B(o, R)$. The parameter $l$ counts the number of non-cluster nodes in the recursion stack we have encountered so far. Recall that we stop as soon as we encounter $K$ of them.
    \item \textsc{ProcessBall}($P$, $r_1$, $r_2$, $o$, $R$, $l$) builds the data structure for a dataset 
    $P$ lying inside the ball $B(o, R)$, assuming we need to solve ANN with distance thresholds $r_1$ and $r_2$.
    Unlike \textsc{ProcessSphere}, here queries can be arbitrary. The parameter $l$ has the same meaning as in \textsc{ProcessSphere}.
    \item \textsc{Project}($R_1$, $R_2$, $r$) is an auxiliary
      function allowing us to project points on a ball to very thin annuli. Suppose we have two spheres $S_1$ and $S_2$
    with a common center and radii $R_1$ and $R_2$. Suppose there are points $p_1 \in S_1$ and $p_2 \in S_2$
    with $\|p_1 - p_2\| = r$. \textsc{Project}($R_1$,~$R_2$,~$r$) returns the distance between $p_1$ and the
    point $\widetilde{p_2}$ that lies on $S_1$ and is the closest to $p_2$ (see Figure~\ref{project_fig}). This is implemented by a formula as in \cite{AR-optimal}.
\end{itemize}

We now elaborate on the above descriptions of \textsc{ProcessSphere} and \textsc{ProcessBall}, since these are the crucial components of our analysis. We will refer to line number of the pseudocode from Figure~\ref{pseudocode}.

\paragraph{\textsc{ProcessSphere}.} We consider three base cases. 
\begin{enumerate}
\item If $l = K$, we stop and store $P$ explicitly. This corresponds to having reached a leaf in the algorithm from Section~\ref{spherical_sec}. This case is handled in lines 2--4 of Figure~\ref{pseudocode}.
\item If $r_2 \geq 2R$, then we may only store one point, since any point in $P$ is a valid answer to any query made on a sphere of radius $R$ containing $P$. This trivial instance is checked in lines 5--7 of Figure~\ref{pseudocode}.
\item The last case occurs when the algorithm from Section~\ref{spherical_sec} can give the desired point on the time--space trade-off. In this case, we may simply proceed as in the algorithm from Section~\ref{spherical_sec}. We choose $\Ueta, \Qeta > 0$ and $T$ appropriately and build a single level of the tree from Section~\ref{spherical_sec} with $l$ increased by 1. We check for this last condition using~(\ref{alpha_beta_tradeoff}) in line 9 of Figure~\ref{pseudocode}, and if so, we may skip lines 10--18 of Figure~\ref{pseudocode}.
\end{enumerate}
If none of the above three cases apply, we proceed in lines 10--18 of Figure~\ref{pseudocode} by removing the dense components and then handling the pseudo-random remainder. The dense components are clusters of at least $\tau |P|$ points lying in a ball of radius $(\sqrt{2} - \eps)R$ with its center on $\partial B(o, R)$. These balls can
be enclosed by smaller balls of radius
$\widetilde{R} \leq (1 - \Omega(\eps^2)) R$. In each of these smaller balls, we invoke
\textsc{ProcessBall} with the same $l$. Finally, we build a single level of the tree in Section~\ref{spherical_sec} for the remaining pseudo-random points. We pick the appropriate $\Ueta, \Qeta > 0$ and $T$ and recurse on each part with \textsc{ProcessSphere} with $l$ increased by $1$.

\paragraph{\textsc{ProcessBall}.}
Similarly to \textsc{ProcessSphere}, if $r_1 + 2R \leq r_2$,
then any point from $B(o, R)$ is a valid answer to any query in $B(o, R + r_2)$. We handle this trivial instance in lines 25--27 of Figure~\ref{pseudocode}.

If we are not in the trivial setting above, we reduce the ball to the spherical case via a discretization of the ball $B(o, R)$ into thin annuli of width $\delta r_1$. First, we round all distances from points to $o$ to a
multiple of $\delta r_1$ in line 28 of Figure~\ref{pseudocode}. This rounding can change the distance between any pair of
points by at most $2 \delta r_1$ by the triangle inequality. 

Then, we handle each non-empty annuli separately. In particular, for a fixed annuli at distance $\delta i r_1$ from $o$, a possible query can lie at most a distance $\delta j r_1$ from $o$, where $\delta r_1 |i - j| \leq r_1 + 2\delta r_1$. For each such case, we recursively build a data structure with \textsc{ProcessSphere}. However, when projecting points, the distance thresholds of $r_1$ and $r_2$ change, and this change is computed using \textsc{Project} in lines 34 and 35 of Figure~\ref{pseudocode}.
\vspace{5mm}

Overall, the preprocessing creates a decision tree. The root corresponds to the procedure {\sc Process}, and subsequent nodes
correspond to procedures {\sc ProcessSphere} and {\sc ProcessBall}. We refer to the tree nodes correspondingly, using the
labels in the description of the query algorithm from below.

\paragraph{Query algorithm.}
Consider a query point $q \in \Rbb^d$. We run the query on the
decision tree, starting with the root which executes \textsc{Process}, and applying the following
algorithms depending on the label of the nodes:
\begin{itemize}
\item In \textsc{ProcessSphere} we first recursively query the data
  structures corresponding to the clusters.
Then, we locate $q$ in the spherical caps (with threshold $\Qeta$, like in Section~\ref{spherical_sec}),
and query data structures we built for the corresponding subsets of $P$. When we encounter a node with points stored explicitly, we simply scan the list of points for a possible near neighbor. This happens when $l = K$.
\item In \textsc{ProcessBall}, we first consider the base case, where we just return the stored point if it is
close enough. In general,
we check whether $\|q - o\|_2 \leq R + r_1$. If not, we return with no neighbor, since each dataset point lies within a ball of radius $R$ from $o$, but the query point is at least $R + r_1$ away from $o$.
If $\|q - o \|_2 \leq R + r_1$, we round $q$ so the distance from $o$ to $q$ is a multiple of $\delta r_1$ and enumerate all possible distances from $o$ to the potential near neighbor we are looking for. For each possible distance, we query the corresponding {\sc ProcessSphere} children after projecting $q$ on the sphere with a tentative near neighbor using, \textsc{Project}.
\end{itemize}

\subsection{Setting parameters}
\label{sec:params}

We complete the description of the data structure by setting the remaining parameters. Recall that the dimension is $d =
\Theta( \log n \cdot \log \log n)$.  We set $\eps, \delta, \tau$ as
follows:
\begin{itemize}
    \item $\eps = \frac{1}{\log \log \log n}$;
    \item $\delta = \exp\bigl(-(\log \log \log n)^C\bigr)$;
    \item $\tau = \exp\bigl(-\log^{2/3} n\bigr)$,
\end{itemize}
where $C$ is a sufficiently large positive constant.

Now we specify how to set $\Ueta, \Qeta > 0$ and $T$ for each pseudo-random remainder. The idea will be to try to replicate the parameter settings of Section~\ref{par_set} corresponding to the random instance. The important parameter will be $r^*$, which acts as the ``effective'' $r_2$. In the case that $r_2 \geq \sqrt{2} R$, then  we have more flexibility than in the random setting, so we let $r^* = r_2$. In the case that $r_2 < \sqrt{2} R$, then we let $r^* = \sqrt{2} R$. In particular, we let 
\[ T = \dfrac{100}{G(r_1/R, \Ueta, \Qeta)} \]
in order to achieve a constant probability of success. Then we let $\Ueta$ and $\Qeta$ such that
\begin{itemize}
\item $F(\Ueta) / G(r_1/R, \Ueta, \Qeta) \leq n^{(\rho_u + o(1)) / K}$
\item $F(\Qeta) / G(r_1/R, \Ueta, \Qeta) \leq n^{(\rho_q + o(1)) / K} $
\item $G(r^*/R, \Ueta, \Qeta) / G(r_1/R, \Ueta, \Qeta) \leq n^{(\rho_q - 1 + o(1)) / K}$
\end{itemize}
which correspond to the parameter settings achieving the tradeoff of Section~\ref{par_set}.

A crucial relation between the parameters is that $\tau$ should be much smaller than $n^{-1/K} = 2^{-\sqrt{\log n}}$. This implies that the ``large distance'' is effectively
equal to $\sqrt{2} R$, at least for the sake of a single step of the random partition.

We collect some basic facts from the data structure which will be useful for the analysis. These facts follow trivially from the pseudocode in Figure~\ref{pseudocode}.
\begin{itemize}

\item \textsc{Process} is called once at the beginning and has one child corresponding to one call to \textsc{ProcessSphere}. In the analysis, we will disregard this node. \textsc{Process} does not take up any significant space or time. Thus, we refer to the root of the tree as the first call to \textsc{ProcessSphere}. 

\item The children to \textsc{ProcessSphere} may contain at most $\frac{1}{\tau}$ many calls to \textsc{ProcessBall}, corresponding to cluster nodes, and $T$ calls to \textsc{ProcessSphere}. Each \textsc{ProcessBall} call of \textsc{ProcessSphere} handles a disjoint subset of the dataset. Points can be replicated in the pseudo-random remainder, when a point lies in the intersection of two or more caps. 

\item \textsc{ProcessBall} has many children, all of which are \textsc{ProcessSphere} which do not increment $l$. Each of these children corresponds to a call on a specific annulus of width $\delta r_1$ around the center as well as possible distance for a query. For each annulus, there are at most $\frac{2}{\delta} + 4$ notable distances; after rounding by $\delta r_1$, a valid query can be at most $r_1+2\delta r_1$ away from a particular annulus in both directions, thus, each point gets duplicated at most $\frac{2}{\delta} + 4$ many times. 

\item For each possible point $p \in P$, we may consider the subtree of nodes which process that particular point. We make the distinction between two kinds of calls to \textsc{ProcessSphere}: calls where $p$ lies in a dense cluster, and calls where $p$ lies in a pseudo-random remainder. If $p$ lies in a dense cluster, $l$ will not be incremented; if $p$ lies in the pseudo-random remainder, $l$ will be incremented. The point $p$ may be processed by various rounds of calls to \textsc{ProcessBall} and \textsc{ProcessSphere} without incrementing $l$; however, there will be a moment when $p$ is not in a dense cluster and will be part of the pseudo-random remainder. In that setting, $p$ will be processed by a call to \textsc{ProcessSphere} which increments $l$.
\end{itemize}

\begin{figure}[p!]
{\footnotesize
\begin{algorithmic}[1]
    \Function{ProcessSphere}{$P$, $r_1$, $r_2$, $o$, $R$, $l$}
        \If{$l = K$} \Comment{base case 1}
        \State store $P$ explicitly 
        \State \Return
        \EndIf
        \If{$r_2 \geq 2R$} \Comment{base case 2}
            \State store any point from $P$
            \State \Return
        \EndIf
        \State $r^* \gets r_2$
        \If{$\left(1 - \alpha\left(\frac{r_1}{R}\right)\alpha\left(\frac{r_2}{R}\right)\right) \sqrt{\rho_q} +
        \left(\alpha\left(\frac{r_1}{R}\right) - \alpha\left(\frac{r_2}{R}\right)\right) \sqrt{\rho_u}
        < 
        \beta\left(\frac{r_1}{R}\right)\beta\left(\frac{r_2}{R}\right)$}  \Comment{base case 3}
        \State $m \gets |P|$
        \State $\widehat{R} \gets (\sqrt{2} - \eps) R$
        \While{$\exists \, x \in \partial B(o, R): |B(x, \widehat{R}) \cap P| \geq \tau m$
               } \Comment{remove dense clusters}
            \State $B(\widetilde{o}, \widetilde{R}) \gets $ the {\sc seb} for
            $P \cap B(x, \widehat{R})$
            \State \Call{ProcessBall}{$P \cap B(x,
                   \widehat{R})$, $r_1$, $r_2$,
                $\widetilde{o}$,
                $\widetilde{R}$, $l$}
            \State $P \gets P \setminus B(x,
                   \widehat{R})$
        \EndWhile
        \State $r^* \gets \sqrt{2} R$
        \EndIf
        \State choose $\Ueta$ and $\Qeta$ such that: \Comment{data independent portion}
        \begin{itemize}
        \item $F(\Ueta) / G(r_1/R, \Ueta, \Qeta) \leq n^{(\rho_u+o(1)) / K }$;
        \item $F(\Qeta) / G(r_1/R, \Ueta, \Qeta) \leq n^{(\rho_q + o(1)) / K }$;
        \item $G(r^*/R, \Ueta, \Qeta) / G(r_1/R, \Ueta, \Qeta) \leq n^{(\rho_q - 1 + o(1)) / K}$.
        \end{itemize}
        \State $T \gets 100 / G(r_1/R, \Ueta, \Qeta)$
        \For{$i \gets 1 \ldots T$} 
        \State sample $z \sim N(0, 1)^d$
        \State $P' \gets \{p \in P \mid \langle z, p \rangle \geq \Ueta R\}$
        \If{$P' \ne \emptyset$}
        \State \Call{ProcessSphere}{$P'$, $r_1$, $r_2$, $o$, $R$, $l + 1$}
        \EndIf
        \EndFor
    \EndFunction
    \Function{ProcessBall}{$P$, $r_1$, $r_2$, $o$, $R$, $l$}
        \If{$r_1 + 2R \leq r_2$} \Comment{trivial instance of \textsc{ProcessBall}}
            \State store any point from $P$
            \State \Return
        \EndIf
        \State $P \gets \{o + \delta r_1 \lceil \frac{\|p - o\|}{\delta r_1}\rceil \cdot \frac{p - o}{\|p - o\|} \mid p \in P\}$
        \For{$i \gets 1 \ldots \lceil \frac{R}{\delta r_1} \rceil$}
            \State $\widetilde{P} \gets \{p \in P \colon \|p - o\| = \delta i r_1 \}$ 
            \If{$\widetilde{P} \ne \emptyset$}
                \For{$j \gets 1 \ldots \lceil
                    \frac{R + r_1 + 2 \delta r_1}{\delta r_1} \rceil$}
                    \If{$\delta |i - j| \leq r_1 + 2 \delta r_1$}
                        \State $\widetilde{r_1} \gets
                        \mbox{\Call{Project}{$\delta i r_1$, $\delta j r_1$,
                        $r_1 + 2 \delta r_1$}}$ \Comment{computing $\widetilde{r_1}$ and $\widetilde{r_2}$ for projected instance}
                        \State $\widetilde{r_2} \gets
                        \mbox{\Call{Project}{$\delta i r_1$, $\delta j r_1$,
                        $r_2 - 2 \delta r_1$}}$
                        \State
                        \Call{ProcessSphere}{$\widetilde{P}$, $\widetilde{r_1}$, $\widetilde{r_2}$, $o$, $\delta i r_1$, $l$}
                    \EndIf
                \EndFor
            \EndIf
        \EndFor
    \EndFunction
    \Function{Project}{$R_1$, $R_2$, $r$}
        \State \Return $\sqrt{R_1 (r^2 - (R_1 - R_2)^2) / R_2}$
    \EndFunction
\end{algorithmic}
}
\caption{Pseudocode of the data structure (\textsc{seb} stands for \emph{smallest enclosing ball})}
\label{pseudocode}
\end{figure}

\subsection{Analysis}

\begin{lemma}
\label{invariants}
The following invariants hold.
\begin{itemize}
\item At any moment one has $\frac{r_2}{r_1} \geq c \cdot (1 - o(1))$ and $r_2 \leq c \cdot ( 1 + o(1) )$.
\item At any moment the number of calls to \textsc{ProcessBall} in the recursion stack is at most
$\widetilde{O}(\log \log n)$.
\end{itemize}
\end{lemma}

\begin{proof}
Our proof will proceed by keeping track of two quantities, $\gamma$ and $\xi$ as the algorithm proceeds down the tree. We will be able to analyze how these values change as the algorithm executes the subroutines \textsc{ProcessSphere} and \textsc{ProcessBall}. We will then combine these two measures to get a potential function which always increases by a multiplicative factor of $(1+\Omega(\eps^2))$. By giving overall bounds on $\gamma$ and $\xi$, we will deduce an upper bound on the depth of the tree.
For any particular node of the tree $v$ (where $v$ may correspond to a call to \textsc{ProcessSphere} or \textsc{ProcessBall}), we let
\[ \gamma_v = \dfrac{r_2^2}{r_1^2} \qquad \text{and} \qquad \xi_v = \dfrac{r_2^2}{R^2} \]
where the values of $r_1, r_2$, and $R$ are given by the procedure call at $v$. We will often refer to $\gamma_v$ and $\xi_v$ as $\gamma$ and $\xi$, respectively, when there is no confusion. Additionally, we will often refer to how $\gamma$ and $\xi$ changes; in particular, if $\tilde{v}$ is a child of $v$, then we let $\tilde{\gamma}$ and $\tilde{\xi}$ be the values of $\gamma_{\tilde{v}}$ and $\xi_{\tilde{v}}$. 

\begin{claim}
\label{cl:etachange}
Initially, $\gamma = c^2$, and it only changes when \textsc{ProcessBall} calls \textsc{ProcessSphere}. Letting $\Delta R = \delta r_1 |i - j|$, there are two cases:
\begin{itemize}
\item If $ 0\leq \dfrac{\Delta R^2}{r_1^2} \leq \frac{12\delta}{\lambda}$, then $\dfrac{\tilde{\gamma}}{\gamma} \geq 1 - 24 \delta$.
\item If $\dfrac{\Delta R^2}{r_2^2} \geq \frac{12\delta}{\lambda}$, then $\dfrac{\tilde{\gamma}}{\gamma} \geq 1 + \dfrac{\Delta R^2}{r_1^2} \cdot \frac{\lambda}{2} - 6 \delta$. 
\end{itemize}
for $\lambda = 1 - (\frac{2}{c+1})^2 > 0$.
\end{claim}

Note that initially, $\gamma = c^2$ since $r_1 = 1$ and $r_2 = c$. Now, $\tilde{\gamma}$ changes when \textsc{ProcessBall}$(P, r_1 ,r_2, o, R, l)$ calls \textsc{ProcessSphere}$(\widetilde{P}, \widetilde{r_1}, \widetilde{r_2}, o, \delta i r_1, l)$ in line 36 of Figure~\ref{pseudocode}. When this occurs: 
\begin{align*}
\dfrac{\tilde{\gamma}}{\gamma} &= \dfrac{\textsc{Project}(\delta i r_1, \delta j r_1, r_2 - 2\delta r_1)^2 / \textsc{Project}(\delta i r_1, \delta jr_1, r_1 + 2\delta r_1)^2}{r_2^2 / r_1^2} \\
                                &= \dfrac{(r_2 - 2\delta r_1)^2 - \Delta R^2}{r_2^2} \cdot \dfrac{r_1^2}{(r_1 + 2\delta r_1)^2 - \Delta R^2} \\
                                &= \dfrac{(1 - \frac{2\delta r_1}{r_2})^2 - \frac{\Delta R^2}{r_2^2}}{(1 + 2\delta)^2 - \frac{\Delta R^2}{r_1^2}} \geq 1 + \dfrac{\Delta R^2 (\frac{1}{r_1^2} - \frac{1}{r_2^2}) -12 \delta }{(1+2\delta)^2 - \frac{\Delta R^2}{r_1^2}} \\
                                &= 1 + \dfrac{\frac{\Delta R^2}{r_1^2} \cdot \lambda - 12\delta}{(1 + 2\delta)^2 - \frac{\Delta R^2}{r_1^2}}
\end{align*}
assuming that $r_1 (\frac{c+1}{2}) \leq r_2$ (we will actually show the much tighter bound of $r_1 \cdot c \cdot (1 - o(1)) \leq r_2$ toward the end of the proof) and setting $\lambda = 1 - (\frac{2}{c+1})^2$, where $\lambda \in (0, 1)$. Note that the denominator in the expression above is non-negative since $\frac{\Delta R^2}{r_1^2} \leq (1 + 2\delta)^2$. Now, consider two cases:
\begin{itemize}
\item Case 1: $0 \leq \frac{\Delta R^2}{r_1^2} \leq \frac{12\delta}{\lambda}$. In this case, we have:
\begin{align*}
\dfrac{\tilde{\gamma}}{\gamma} &\geq 1 + \dfrac{\frac{\Delta R^2}{r_1^2} \cdot \lambda - 12 \delta}{(1+2\delta)^2 - \frac{\Delta R^2}{r_1^2}} \\
                                                 &\geq 1 - \dfrac{12 \delta}{(1 + 2\delta)^2 - \frac{\Delta R^2}{r_1^2}} \geq 1 - \dfrac{12 \delta}{(1/2)}.
\end{align*} 
Thus, the multiplicative decrease is at most $(1 - 24 \delta)$ since $\delta = o(1)$. 
\item Case 2: $\frac{\Delta R^2}{r_1^2} \geq \frac{12 \delta}{\lambda}$. In this case: 
\begin{align*}
\dfrac{\tilde{\gamma}}{\gamma} &\geq 1 + \dfrac{\frac{\Delta R^2}{r_1^2} \cdot \lambda - 12 \delta}{(1+2\delta)^2 - \frac{\Delta R^2}{r_1^2}} \\
                                                &\geq 1 + \dfrac{\frac{\Delta R^2}{r_1^2} \cdot \lambda - 12 \delta}{2} \\
                                                &= 1 + \dfrac{\Delta R^2}{r_1^2} \cdot \frac{\lambda}{2} - 6\delta.
\end{align*}
\end{itemize}

\begin{claim}
\label{cl:xichange}
Initially, $\xi \geq \widetilde{\Omega}\left(\dfrac{c^2}{\log^4 n}\right)$. $\xi$ changes only when \textsc{ProcessBall} calls \textsc{ProcressSphere}, or vice-versa. When \textsc{ProcessBall} calls \textsc{ProcessSphere}, and some later \textsc{ProcessSphere} calls \textsc{ProcessBall}, letting $\Delta R = \delta r_1 |i - j|$: 
\[ \dfrac{\tilde{\xi}}{\xi} \geq \left(1 + \Omega(\eps^2)\right) \left( (1 - 2\delta)^2 - \frac{\Delta R^2}{r_1^2}(1 - \lambda) \right) \left( \dfrac{1}{1 + \frac{\Delta R}{R}}\right),\]
for $\lambda = 1 - (\frac{2}{c+1})^2 > 0$.
\end{claim}

The relevant procedure calls in Claim~\ref{cl:xichange} are: 
\begin{enumerate}
\item \textsc{ProcessBall}$(P, r_1, r_2, o, R, l)$ calls \textsc{ProcessSphere}$(\widetilde{P},\widetilde{r_1}, \widetilde{r_2}, o, \delta i r_1, l )$ in line 36 of Figure~\ref{pseudocode}.
\item After possibly some string of calls to \textsc{ProcessSphere}, \textsc{ProcessSphere}$(P', \widetilde{r_1}, \widetilde{r_2}, o, \delta i r_1, l')$ calls \textsc{ProcessBall}$(P' \cap B(\widetilde{o}, \widetilde{R}), \widetilde{r_1}, \widetilde{r_2}, \widetilde{o}, l')$ in line 14 of Figure~\ref{pseudocode}.
\end{enumerate}
Since both calls to \textsc{ProcessBall} identified above have no \textsc{ProcessBall} calls in their path in the tree, we have the following relationships between the parameters:
\begin{itemize}
\item $\widetilde{r_1} = \textsc{Process}(\delta i r_1, \delta j r_1, r_1 + 2\delta)$,
\item $\widetilde{r_2} = \textsc{Process}(\delta i r_1, \delta j r_1, r_2 - 2\delta)$,
\item $\widetilde{R} \leq (1 - \Omega(\eps^2)) \cdot  \delta i r_1$,
\end{itemize}

Using these setting of parameters, 
\begin{align*}
\dfrac{\tilde{\xi}}{\xi} &= \left(1 + \Omega(\eps^2)\right)^2 \cdot \dfrac{\textsc{Process}(\delta i r_1, \delta j r_1, r_2 - 2\delta r_1)^2 / \delta^2 i^2 r_1^2}{r_2^2 / R^2} \\
                                  &= \left( 1 + \Omega(\eps^2)\right)^2 \left( \left(1 - \frac{2\delta r_1}{r_2}\right)^2 - \frac{\Delta R^2}{r_2^2} \right) \cdot \dfrac{R^2}{\delta^2 i j r_1^2} \\
                                  &\geq \left( 1 + \Omega(\eps^2)\right) \left( (1 - 2\delta)^2 - \frac{\Delta R^2}{r_1^2} \cdot (1 - \lambda) \right) \left( \dfrac{1}{1 + \frac{\Delta R}{R}}\right),
\end{align*}
where in the last step, we used the fact that $r_1 (\frac{c+1}{2}) \leq r_2$, and that $\delta i r_1 \leq R$ and $\delta j r_1 \leq R + \Delta R$. Note that the lower bound is always positive since $\frac{\Delta R^2}{r_1^2} \leq 1 + 2\delta$, $\delta = o(1)$, and $\lambda \in (0, 1)$ is some constant. 

We consider the following potential function:
\[ \Phi = \gamma^M \cdot \xi, \]
and we set $M = \frac{800}{\sqrt{ 24 \cdot \lambda \cdot \delta}}$. 
\begin{claim}
In every iteration of \textsc{ProcessBall} calling \textsc{ProcessSphere} which at some point calls \textsc{ProcessBall} again, $\Phi$ increases by a multiplicative factor of $1 + \Omega(\eps^2)$. 
\end{claim}
We simply compute the multiplicative change in $\Phi$ by using Claim~\ref{cl:etachange} and Claim~\ref{cl:xichange}. We will first apply the first case of Claim~\ref{cl:etachange}, where $0 < \frac{\Delta R^2}{r_1^2} \leq \frac{24\delta}{\lambda}$. 
\begin{align*}
\dfrac{\widetilde{\Phi}}{\Phi} &\geq \left(1 - 24 \delta\right)^M \cdot \left( 1 + \Omega(\eps^2) \right) \cdot \left( (1 - 2\delta)^2 - \frac{\Delta R^2}{r_1^2} (1 - \lambda)  \right) \cdot \left( \dfrac{1}{1 + \frac{\Delta R}{R}} \right) \\
                                   &\geq \left(1 + \Omega(\eps^2) \right) \cdot \left(1 - 24 \delta M - 4\delta - \frac{\Delta R^2}{r_1^2} - \frac{\Delta R}{R}\right) \\
                                   &\geq \left(1 + \Omega(\eps^2) \right) \cdot \left(1 - 24 \delta M - 4\delta - \frac{24 \delta}{\lambda} - \frac{\sqrt{96 \cdot \delta}}{\sqrt{\lambda}} \right),
\end{align*}
where the third inequality, we used $\frac{\Delta R^2}{r_1^2} \leq \frac{24 \delta}{\lambda}$ and $r_1 \leq r_2 \leq 2R$ by line 5 of Figure~\ref{pseudocode}. Finally, we note that $\eps^2 \gg 24\delta M - 4\delta - \frac{24\delta}{\lambda} - \frac{\sqrt{96 \delta}}{\sqrt{\lambda}}$, so in this case,
\[ \dfrac{\widetilde{\Phi}}{\Phi} \geq \left(1 + \Omega(\eps^2) \right). \]

We now proceed to the second case, when $\frac{\Delta R^2}{r_2^2} > \frac{24\delta}{\lambda}$. Using case 2 of Claim~\ref{cl:etachange}, we have
\begin{align}
\dfrac{\tilde{\Phi}}{\Phi} &\geq \left(1 + \dfrac{\Delta R^2}{r_1^2} \cdot \frac{\lambda}{4} \right)^M \cdot \left(1 + \Omega(\eps^2) \right) \cdot \left( (1 - 2\delta)^2 - \frac{\Delta R^2}{r_1^2} (1 - \lambda)\right) \cdot \left( \dfrac{1}{1 + \frac{\Delta R}{R}}\right).
\end{align}
We claim the above expression is at least $1 + \Omega(\eps^2)$. This follows from three observations:
\begin{itemize}
\item $\frac{\Delta R^2}{r_1^2} \geq \frac{24 \delta}{\lambda}$ implies that $\frac{\sqrt{\lambda}}{\sqrt{24 \cdot \delta}} \geq \frac{r_1}{\Delta R}$. 
\item Since $r_1 \leq r_2 \leq 2R$ (by line 5 of Figure~\ref{pseudocode}), $\frac{2}{r_1} \geq \frac{1}{R}$, so $\frac{\Delta R^2}{r_1^2} \cdot \frac{2\sqrt{\lambda}}{\sqrt{24 \cdot \delta}} \geq \frac{2\Delta R}{r_1} \geq \frac{\Delta R}{R}$.
\item Thus, if $M = \frac{800}{\sqrt{ 24 \cdot \lambda \cdot \delta}}$,
\[ \dfrac{\Delta R^2}{r_1^2} \cdot \frac{\lambda}{4} \cdot M \geq 100 \cdot \dfrac{\Delta R^2}{r_1^2} \cdot \dfrac{2 \cdot \sqrt{\lambda}}{\sqrt{24 \cdot \delta}} \geq 100 \cdot \dfrac{\Delta R}{R}.\]
\end{itemize}
Furthermore, $\frac{\Delta R^2}{r_1^2} \cdot \frac{\lambda}{4} \cdot M \gg 4\delta + \frac{\Delta R^2}{r_1^2}$, which means that in this case,
\[ \dfrac{\widetilde{\Phi}}{\Phi} \geq \left(1 + \Omega(\eps^2) \right). \]

Having lower bounded the multiplicative increase in $\Phi$, we note that initially, 
\[ \Phi_0 = \widetilde{\Omega}\left(\dfrac{c^{2M + 2}}{\log^4 n}\right). \]
\begin{claim}
At all moments in the algorithm $\gamma \leq \log n$.
\end{claim}
Note that before reaching $\frac{r_2^2}{r_1^2} \geq \log n$, line 9 of Figure~\ref{pseudocode} will always evaluate to false, and the algorithm will continue to proceed in a data-independent fashion without further changes to $r_1$, $r_2$ or $R$. Another way to see this is that when $\frac{r_2}{r_1} \geq \sqrt{\log n}$, then the curve in Figure~\ref{trade_off_plot} corresponding to \cite{AI-CACM} will give a data structure with runtime $n^{o(1)}$ and space $n^{1 + o(1)}$. 

Additionally, line 5 of Figure~\ref{pseudocode} enforces that all moments in the algorithm, $\xi \leq 4$. Thus, at all moments, 
\[ \Phi \leq O(\log^{M} n). \]
Thus, the number of times that \textsc{ProcessBall} appears in the stack is $\widetilde{O}(\log \log n)$. We will now show the final part of the proof which we stated earlier:
\begin{claim}
For the first $\widetilde{O}\left( \log \log n \right)$ many iterations, 
\[ r_1 \cdot c \cdot \left(1 - o(1) \right) \leq r_2. \]
\end{claim} 
Note that showing this will imply $\eta \geq \left(\frac{c+1}{2}\right)^2$. From Claim~\ref{cl:etachange},
$\eta \geq c^2 \left( 1 - 24\delta \right)^{N}$, where $N = \widetilde{O}\left(\log \log n\right)$, which is in fact, always at most $c^2 (1 - o(1))$. In order to show that $r_2 \leq c \cdot (1 + o(1))$, note that $r_2$ only increases by a multiplicative factor of $(1 + 2\delta)$ in each call of \textsc{ProcessBall}. This finishes the proof of all invariants. 
\end{proof}

\begin{lemma}
During the algorithm we can always be able to choose $\Ueta$ and $\Qeta$ such that:
\begin{itemize}
\item $F(\Ueta) / G(r_1/R, \Ueta, \Qeta) \leq n^{(\rho_u+o(1)) / K}$;
\item $F(\Qeta) / G(r_1/R, \Ueta, \Qeta) \leq n^{(\rho_q +o(1))/ K}$;
\item $G(r^*/R, \Ueta, \Qeta) / G(r_1/R, \Ueta, \Qeta) \leq n^{(\rho_q - 1 + o(1)) / K}$.
\end{itemize}
\end{lemma}

\begin{proof}
We will focus on the the part of \textsc{ProcessSphere} where we find settings for $\Ueta$ and $\Qeta$. There are two important cases:
\begin{itemize}
\item $r^* = r_2$. This happens when the third ``if'' statement evaluates to false. In other words, we have that
\begin{align}
\left(1 - \alpha\left(\frac{r_1}{R}\right) \alpha\left(\frac{r_2}{R}\right)\right) \sqrt{\rho_q} + \left(\alpha\left(\frac{r_1}{R}\right) - \alpha\left(\frac{r_2}{R}\right)\right) \sqrt{\rho_u} \geq \beta\left(\frac{r_1}{R}\right) \beta\left(\frac{r_2}{R}\right). \label{eq:tradeoff_in_sphere}
\end{align}
Since in a call to \textsc{ProcessSphere}, all points are on the surface of a sphere of radius $R$, the expression corresponds to the expression from Theorem~\ref{main_thm_spherical}. Thus, as described in Section~\ref{par_set}, we can set $\Ueta$ and $\Qeta$ to satisfy the three conditions. 
\item $r^* = \sqrt{2}R$. This happens when the third ``if'' statement evaluates to true. We have by Lemma~\ref{invariants} $\frac{r_2}{r_1}\geq c - o(1)$. Since (\ref{eq:tradeoff_in_sphere}) does not hold, thus $r_2 < \sqrt{2} R$. Hence, $r_1 \leq \frac{\sqrt{2}R}{c} - o(1)$. If this is the case since $r_1 \leq r^* / c - o(1)$, we are instantiating parameters as in Subsection~\ref{par_set} where $r=\frac{r_1}{R}$ and $cr = \frac{r^*}{R}$.
\end{itemize}
\end{proof}

\begin{lemma}
The probability of success of the data structure is at least $0.9$.
\end{lemma}
\begin{proof}
In all the cases, except for the handling of the pseudo-random remainder, the data structure is deterministic. Therefore, the proof follows in exactly the same way as Lemma~\ref{lem:succ-prob}. In this case, we also have at each step that $T = \frac{100}{G(r_1/R, \Ueta, \Qeta)}$, and the induction is over the number of times we handle the pseudo-random remainder.
\end{proof}

\begin{lemma}
The total space the data structure occupies is at most $n^{1 + \rho_u + o(1)}$ in expectation.
\end{lemma}

\begin{proof}
We will prove that the total number of explicitly stored points (when $l = K$) is at most $n^{1 + \rho_u + o(1)}$. We will count the contribution from each point separately, and use linearity of expectation to sum up the contributions. In particular, for a point $p \in P_0$, we want to count the number of lists where $p$ appears in the data structure. Each root to leaf path of the tree has at most $K$ calls to \textsc{ProcessSphere} which increment $l$, and at most $\widetilde{O}\left( \log \log n \right)$ calls to \textsc{ProcessBall}, and thus $\widetilde{O}\left( \log \log n\right)$ calls to \textsc{ProcessSphere} which do not increment $l$. Thus, once we count the number of lists, we may multiply by $K + \widetilde{O}\left( \log \log n\right) = n^{o(1)}$ to count the size of the whole tree. 

For each point, we will consider the subtree of the data structure where the point was processed. In particular, we may consider the tree corresponding to calls to \textsc{ProcessSphere} and \textsc{ProcessBall} which process $p$. As discussed briefly in Section~\ref{sec:params}, we distinguish between calls to \textsc{ProcessSphere} which contain $p$ in a dense cluster, and calls to \textsc{ProcessSphere} which contain $p$ in the pseudo-random remainder. We increment $l$ only when $p$ lies in the pseudo-random remainder.

\begin{claim}
It suffices to consider the data structure where each node is a function call to \textsc{ProcessSphere} which increments $l$, i.e., when $p$ lies in the pseudo-random remainder, since the total amount of duplication of points corresponding to other nodes is $n^{o(1)}$.
\end{claim}
We will account for the duplication of points in calls to \textsc{ProcessBall} and calls to \textsc{ProcessSphere} which do not increment $l$. Consider the first node $v$ in a path from the root which does not increment $l$, this corresponds to a call to \textsc{ProcessSphere} which had $p$ in some dense cluster. Consider the subtree consisting of descendants of $v$ where the leaves correspond to the first occurrence of \textsc{ProcessSphere} which increments $l$. We claim that every internal node of the tree corresponds to alternating calls to \textsc{ProcessBall} and \textsc{ProcessSphere} which do not increment $l$. Note that calls to \textsc{ProcessSphere} which do not increment $l$ never replicate $p$. Each call to \textsc{ProcessBall} replicates $p$ in $b := \frac{2r_1(1 + 2\delta)}{\delta}$ many times. Since $r_1 \leq r_2 \leq c + o(1)$ by Lemma~\ref{invariants}, $b = O(\delta^{-1})$. We may consider contracting the tree and at edge, multiplying by the number of times we encounter \textsc{ProcessBall}. 

Note that $p$ lies in a dense cluster if and only if it does not lie in the pseudo-random remainder. Thus, our contracted tree looks like a tree of $K$ levels, each corresponding to a call to \textsc{ProcessSphere} which contained $p$ in the pseudo-random remainder. 

The number of children of some nodes may be different; however, the number of times \textsc{ProcessBall} is called in each branch of computation is $U := \tilde{O}\left(\log \log n\right)$, the total amount of duplication of points due to \textsc{ProcessBall} is at most $b^U = n^{o(1)}$.  Now, the subtree of nodes processing $p$ contains $K$ levels with each $T$ children, exactly like the data structure for Section~\ref{spherical_sec}.

\begin{claim}
A node $v$ corresponding to \textsc{ProcessSphere}$(P, r_1,r_2, o, R, l)$ has in expectation, $p$ appearing in $n^{((K - l) \rho_u+o(1))/K}$ many lists in the subtree of $v$.
\end{claim}
The proof is an induction over the value of $l$ in a particular node. For our base case, consider some node $v$ corresponding to a function call of \textsc{ProcessSphere} which is a leaf, so $l = K$, in this case, each point is only stored at most once, so the claim holds. 

Suppose for the inductive assumption the claim holds for some $l$, then for a particular node at level $l-1$, consider the point when $p$ was part of the pseudo-random remainder. In this case, $p$ is duplicated in 
\[ T \cdot F(\Ueta) = \frac{100 \cdot F(\Ueta)}{G(r_1/R, \Ueta, \Qeta)} \leq n^{(\rho_u + o(1))/ K}\]
many children, and in each child, the point appears $n^{((K - l) \rho_u +o(1))/ K }$ many times. Therefore, in a node $v$, $p$ appears in $n^{((K-l + 1) \rho_u + o(1)) / K}$ many list in its subtree. Letting $l = 0$ for the root gives the desired outcome. 
\end{proof}

\begin{lemma}
The expected query time is at most $n^{\rho_q + o(1)}$.
\end{lemma}

\begin{proof}
We need to bound the expected number of nodes we traverse
as well as the number of points we enumerate for nodes with $l = K$.

We first bound the number of nodes we traverse. Let $A(u, l)$ be an upper bound on the expected number of visited nodes when we start
in a \textsc{ProcessSphere} node such that there are $u$ \textsc{ProcessBall} nodes in the stack and $l$ non-cluster nodes.
By Lemma~\ref{invariants},
$$u \leq U := \widetilde{O}\left(\log \log n\right),$$ 
and from the description of the algorithm,
we have $l \leq K$. We will prove $A(0, 0) \leq n^{\rho_q + o(1)}$, which corresponds to the expected number of nodes we touch starting from the root.

We claim
\begin{equation}
\label{recurrence_a}
A(u, l) \leq \exp(\log^{2/3 + o(1)} n) \cdot A(u + 1, l) + n^{(\rho_q +o(1)) / K } \cdot A(u, l + 1).
\end{equation}

There are at most $1 / \tau = \exp(\log^{2/3} n)$ cluster nodes, and in each node, we recurse on $\frac{2r_1(1+2\delta)}{\delta} = \exp(\log^{o(1)} n)$ possible annuli with calls to \textsc{ProcessSphere} nodes where $u$ increased by $1$ and $l$ remains the same. On the other hand, there are
$$
T \cdot F(\Qeta) = \frac{100 \cdot F(\Qeta)}{G(r_1/R, \Ueta, \Qeta)} \leq n^{(\rho_q + o(1))/ K}
$$
caps, where the query falls, in expectation. Each calls \textsc{ProcessSphere} where $u$ remains the same and $l$ increased by $1$.

Solving~(\ref{recurrence_a}):
$$
A(0, 0) \leq \binom{U + K}{K} \exp(U \cdot \log^{2/3 + o(1)} n) \cdot n^{\rho_q + o(1)} \leq n^{\rho_q + o(1)}.
$$

We now give an upper bound on the number of points the query algorithm will test at level $K$. Let $B(u, l)$ be an upper bound on the expected fraction of the dataset in the current node that the query algorithm will eventually test at level $K$ (where we count multiplicities). $u$ and $l$ have the same meaning as discussed above.

We claim
\[ B(u, l) \leq \frac{1}{\tau} \cdot B(u+1, l) + n^{(\rho_q - 1 + o(1))/K} \cdot B(u, l+1) \]
The first term comes from recursing down dense clusters. The second term is a bit more subtle. In particular, suppose $r_2 = r^*$, then the expected fraction of points is
\begin{align*}
T \cdot G(r_2/R, \Ueta, \Qeta) \cdot B(u, l+1) &= \dfrac{100 \cdot G(r_2/R, \Ueta, \Qeta) \cdot B(u, l+1)}{G(r_1/R, \Ueta, \Qeta)}\\
                                                            &\leq n^{(\rho_q - 1 +o(1))/K} \cdot B(u, l+1)
\end{align*}
by the setting of $\Ueta$ and $\Qeta$. On the other hand, there is the other case when $r^* = \sqrt{2} R$, which occurs after having removed some clusters. In that case, consider a particular cap containing the points $\widetilde{P}_i$. For points with distance to the query at most $(\sqrt{2} - \eps) R$, there are at most a $\tau n$ of them. For the far points, $\widetilde{P}_i$ a $G(\sqrt{2} - \eps, \Ueta, \Qeta)$ fraction of the points in expectation. 
\begin{align*}
T \cdot F(\Qeta) \cdot \left( \tau + G(\sqrt{2} - \eps, \Ueta, \Qeta) \right) \cdot B(u, l+1) &= \dfrac{100 \cdot F(\Qeta) \cdot \left( \tau + G(\sqrt{2} - \eps, \Ueta, \Qeta) \right) \cdot B(u, l + 1)}{G(r_1/R, \Ueta, \Qeta)} \\
\\      &\leq \dfrac{100 \cdot F(\Qeta) \cdot G(\sqrt{2}, \Ueta, \Qeta) \cdot B(u, l+1)}{G(r_1/R, \Ueta, \Qeta)} \\
        &\leq n^{(\rho_q - 1 +o(1))/K} \cdot B(u, l+1)
\end{align*}
Where we used that $\tau \ll G(\sqrt{2}, \Ueta, \Qeta)$ and $G(\sqrt{2} - \eps, \Ueta, \Qeta) \leq G(\sqrt{2}, \Ueta, \Qeta) \cdot n^{o(1)/K}$ (since $\eps = o(1)$), and that $r^* = \sqrt{2}R$. Unraveling the recursion, we note that $u \leq U$ and $l \leq K \sim \sqrt{\ln n}$. Additionally, we have that $B(u, K) \leq 1$, since we do not store duplicates in the last level. Therefore,
\[ B(0, 0) \leq \dbinom{U + K}{U} \left(\frac{1}{\tau}\right)^{U} \cdot \left(n^{(\rho_q - 1 +o(1))/K}\right)^K = n^{\rho_q - 1 + o(1)}.\]
\end{proof}


\section{Lower bounds: preliminaries}

We introduce a few techniques and concepts to be used primarily for
our lower bounds. We start by defining the approximate nearest
neighbor search problem.

\begin{definition}
The goal of the $(c,r)$-approximate nearest neighbor problem with failure
probability $\delta$ is to construct a data structure over a set of
points $P\subset \{-1,1\}^d$ supporting the following query: given any
point $q$ such that there exists some $p \in P$ with $\|q - p\|_1 \leq
r$, report some $p' \in P$ where $\|q - p'\|_1 \leq c r$ with
probability at least $1- \delta$.
\end{definition}

\subsection{Graphical Neighbor Search and robust expansion}

We introduce a few definitions from \cite{PTW10} to setup the lower bounds for 
the ANN.

\begin{definition}[\cite{PTW10}]
\label{def:gns}
In the \emph{Graphical Neighbor Search problem} (GNS), we are given a
bipartite graph $G = (U, V, E)$ where the dataset comes from $U$ and
the queries come from $V$. The dataset consists of pairs $P = \{ (p_i,
x_i) \mid p_i \in U, x_i \in \{0, 1\}, i \in [n] \}$. On query $q\in
V$, if there exists a unique $p_i$ with $(p_i, q) \in E$, then we want
to return $x_i$.
\end{definition}

We will sometimes use the GNS problem to prove lower bounds on
$(c,r)$-ANN as follows: we build a GNS graph $G$ by taking $U = V =
\{-1,1\}^d$, and connecting two points $u\in U, v\in V$ iff their Hamming distance
 most $r$ (see details in \cite{PTW10}). We will also
ensure $q$ is not closer than $cr$ to
other points apart from the near neighbor.

\cite{PTW10} showed lower bounds for ANN are intimately tied to the following 
quantity of a metric space.

\begin{definition}[Robust Expansion \cite{PTW10}]
Consider a GNS graph $G=(U,V,E)$, and fix a distribution $e$ on $E\subset
U\times V$, where $\mu$ is the marginal distribution on $U$ and $\eta$ is the
marginal distribution on $V$.  For $\delta, \gamma\in(0,1]$, the \emph{robust expansion}
$\Phi_r(\delta, \gamma)$ is:
$$
\Phi_r(\delta, \gamma)=\min_{A\subset V : \eta(A)\le \delta} \min_{B\subset U :
    \frac{e(A\times B)}{e(A\times V)}\ge \gamma} \frac{\mu(B)}{\eta(A)}.
$$
\end{definition}

\subsection{Locally-decodable codes (LDC)}

Our 2-probe lower bounds uses results on Locally-Decodable
Codes (LDCs). We present the standard definitions and results on LDCs
below, although in Section~\ref{sec:twoProbes}, we will use a weaker definition of LDCs for our 2-query lower bound.

\begin{definition}
\label{def:ldc}
A $(t, \delta, \eps)$ locally-decodable code (LDC) encodes $n$-bit
strings $x\in\{0,1\}^n$ into $m$-bit codewords $C(x)\in\{0,1\}^m$ such
that, for each $i\in[n]$, the bit $x_i$ can be recovered with
probability $\frac{1}{2} + \eps$ while making only $t$ queries into
$C(x)$, even if the codeword is arbitrarily modified (corrupted) in
$\delta m$ bits.
\end{definition}

We will use the following lower bound on the size of the LDCs.

\begin{theorem}[Theorem 4 from \cite{KW2004}]
If $C:\{0, 1\}^n \to \{0, 1\}^m$ is a $(2, \delta, \eps)$-LDC, then
\begin{align*}
m &\geq 2^{\Omega(\delta \eps^2 n)}.
\end{align*}
\end{theorem}


\section{Lower bounds: one-probe data structures} 

\subsection{Robust expansion of the Hamming space}

The goal of this section is to compute tight bounds for the robust
expansion $\Phi_r(\delta, \gamma)$ in the Hamming space of dimension
$d$, as defined in the preliminaries. We use these bounds for all of
our lower bounds in the subsequent sections.

We use the following model for generating dataset points and queries
corresponding to the random instance of Section~\ref{sec:rand_inst_sec}.

\begin{definition}
For any $x \in \{-1, 1\}^n$, $N_\sigma(x)$ is a probability
distribution over $\{-1, 1\}^n$ representing the neighborhood of
$x$. We sample $y \sim N_{\sigma}(x)$ by choosing $y_i \in \{-1, 1\}$ for each
coordinate $i \in [d]$ independently; with probability $\sigma$, we set $y_i = x_i$, and with probability $1- \sigma$, $y_i$ is set uniformly at random.

Given any Boolean function $f:\{-1, 1\}^n \to \R$, the function $T_\sigma f:\{-1, 1\}^n \to \R$ is
\begin{align}
T_\sigma f(x) &= \mathop{\E}_{y \sim N_\sigma(x)} [f(y)]
\end{align}
\end{definition}

In the remainder of this section, will work solely on the Hamming space $V = \{-1, 1\}^d$. We let 
\[ \sigma = 1 - \frac{1}{c} \qquad  \qquad d = \omega(\log n) \] 
and $\mu$ will refer to the uniform distribution over $V$. 

A query is generated as follows: we sample a dataset point $x$ uniformly at random and then generate the query $y$ by sampling $y \sim N_{\sigma}(x)$. From the choice of $\sigma$ and $d$, $d(x, y) \leq \frac{d}{2c}(1 + o(1))$ with high probability. In addition, for every other point in the dataset $x' \neq x$, the pair $(x', y)$ is distributed as two uniformly random points (even though $y \sim N_{\sigma}(x)$, because $x$ is randomly distributed). Therefore, by taking a union-bound over all dataset points, we can conclude that with high probability, $d(x', y) \geq \frac{d}{2}(1 - o(1))$ for each $x' \neq x$. 

Given a query $y$ generated as described above, we know there exists a dataset point $x$ whose distance to the query is $d(x, y) \leq \frac{d}{2c}(1 + o(1))$. Every other dataset point lies at a distance $d(x', y) \geq \frac{d}{2}(1 - o(1))$. Therefore, the two distances are a factor of $c - o(1)$ away. 

The following lemma is the main result of this section, and we will reference this lemma in subsequent sections.

\begin{lemma}[Robust expansion]
\label{lem:robustexpansion}
Consider the Hamming space equipped with the Hamming norm. For any $p, q \in [1, \infty)$ where $(q-1)(p-1) = \sigma^2$, any $\gamma \in [0, 1]$, and $m \geq 1$,
\begin{align*}
\Phi_r\left(\frac{1}{m}, \gamma\right) &\geq \gamma^q m^{1 + \frac{q}{p} - q}.
\end{align*}
\end{lemma}

The robust expansion comes from a straight forward application from small-set expansion. In fact, one can easily prove tight bounds on robust expansion via the following lemma:
\begin{theorem}[Generalized Small-Set Expansion Theorem, \cite{AOBF}]
Let $0 \leq \sigma \leq 1$. Let $A, B \subset \{-1, 1\}^n$ have volumes $\exp(-\frac{a^2}{2})$ and $\exp(-\frac{b^2}{2})$ and assume $0 \leq \sigma a \leq b \leq a$. Then
\[ \Pr_{^{\ \ \ (x, y)}_{ \sigma-\text{correlated}}}[x \in A, y \in B] \leq \exp\left( - \frac{1}{2} \frac{a^2 - 2\sigma a b + b^2}{1 - \sigma^2}\right).\]
\end{theorem}

We compute the robust expansion via an application of the Bonami-Beckner Inequality and H\"{o}lder's inequality. This computation gives us more flexibility with respect to parameters which will become useful in subsequent sections. We now recall the necessary tools. 

\begin{theorem}[Bonami-Beckner Inequality \cite{AOBF}]
Fix $1 \leq p \leq q$ and $0 \leq \sigma \leq \sqrt{(p-1)/(q-1)}$. Any Boolean function $f:\{-1, 1\}^n \rightarrow \R$ satisfies
\begin{align*}
\|T_\sigma f\|_q &\leq \|f\|_p.
\end{align*}
\end{theorem}

\begin{theorem}[H\"{o}lder's Inequality]
Let $f:\{-1, 1\}^n \to \R$ and $g:\{-1, 1\}^n \to \R$ be arbitrary Boolean functions. Fix $s, t \in [1, \infty)$ where $\frac{1}{s} + \frac{1}{t} = 1$. Then
\begin{align*}
\langle f, g \rangle &\leq \|f\|_s \|g\|_t .
\end{align*}
\end{theorem}

We will let $f$ and $g$ be indicator functions for two sets $A$ and $B$, and use a combination of the Bonami-Beckner Inequality and H\"{o}lder's Inequality to lower bound the robust expansion. The operator $T_{\sigma}$ applied to $f$ will measure the neighborhood of set $A$. We compute an upper bound on the correlation of the neighborhood of $A$ and $B$ (referred to as $\gamma$) with respect to the volumes of $A$ and $B$, and the expression will give a lower bound on robust expansion. 

We also need the following lemma.

\begin{lemma}
\label{lem:robustisoperimetry}
Let $p, q \in [1, \infty)$, where $(p-1)(q-1) = \sigma^2$ and $f, g: \{-1, 1\}^d \to \R$ be two Boolean functions. Then
\[ \langle T_\sigma f, g \rangle \leq \|f\|_p \|g\|_q. \]
\end{lemma}

\begin{proof}
We first apply H\"{o}lder's Inequality to split the inner-product into two parts, apply the Bonami-Beckner Inequality to each part.
\[ \langle T_{\sigma} f, f \rangle = \langle T_{\sqrt{\sigma}} f, T_{\sqrt{\sigma}} g \rangle \leq \|T_{\sqrt{\sigma}}f\|_s \|T_{\sqrt{\sigma}}g\|_t. \]
We pick the parameters $s = \dfrac{p - 1}{\sigma} + 1$ and $t = \dfrac{s}{s - 1}$, so $\frac{1}{s} + \frac{1}{t} = 1$. Note that $p \leq s$ because $\sigma < 1$ and $p \geq 1$ because $(p-1)(q-1) = \sigma^2 \leq \sigma$. We have 
\begin{align*} 
q \leq \dfrac{\sigma}{p-1} + 1 = t.
\end{align*}
In addition, 
\begin{align*}
\sqrt{\dfrac{p-1}{s - 1}} &= \sqrt{\sigma}  \qquad  &\sqrt{\dfrac{q - 1}{t - 1}} &= \sqrt{(q-1)(s-1)} = \sqrt{\frac{(q-1)(p - 1)}{\sigma}} = \sqrt{\sigma}.
\end{align*}
We finally apply the Bonami-Beckner Inequality to both norms to obtain
\begin{align*}
\|T_{\sqrt{\sigma}} f\|_s \| T_{\sqrt{\sigma}}g\|_t &\leq \|f\|_p \|g\|_q.
\end{align*}
\end{proof}

We are now ready to prove Lemma~\ref{lem:robustexpansion}.

\begin{proof}[Proof of Lemma~\ref{lem:robustexpansion}]
We use Lemma~\ref{lem:robustisoperimetry} and the definition of robust expansion. For any two sets $A, B \subset V$, let $a = \frac{1}{2^d}|A|$ and $b = \frac{1}{2^d}|B|$ be the measure of set $A$ and $B$ with respect to the uniform distribution. We refer to $\chi_A:\{-1, 1\}^d \to \{0, 1\}$ and $\chi_B:\{-1, 1\}^d \to \{0, 1\}$ as the indicator functions for $A$ and $B$. Then,
\begin{equation} \gamma = \Pr_{x\sim\mu, y \sim N_\sigma(x)} [x \in B \mid y \in A] = \frac{1}{a}\langle T_{\sigma}\chi_A, \chi_B\rangle \leq a^{\frac{1}{p}-1} b^{\frac{1}{q}}. \label{eq:min-gamma}
\end{equation}
Therefore, $\gamma^qa^{q-\frac{q}{p}} \leq b$. Let $A$ and $B$ be the minimizers of $\frac{b}{a}$ satisfying (\ref{eq:min-gamma}) and $a \leq \frac{1}{m}$.
\[ \Phi_r\left(\frac{1}{m}, \gamma\right) = \frac{b}{a} \geq \gamma^qa^{q - \frac{q}{p} - 1} \geq \gamma^q m^{1 + \frac{q}{p} - q}. \]
\end{proof}


\subsection{Lower bounds for one-probe data structures}
\label{sec:oneProbe}

In this section, we prove Theorem~\ref{one_probe_thm}.
Our proof relies on the main result of \cite{PTW10} for the GNS
problem:

\begin{theorem}[Theorem 1.5 \cite{PTW10}]
\label{thm:ptwbound}
There exists an absolute constant $\gamma$ such that the following
holds. Any randomized cell-probe data structure making $t$ probes and using $m$
cells of $w$ bits for a weakly independent instance of
GNS which is correct with probability greater than $\frac{1}{2}$ must
satisfy
\begin{align*}
\dfrac{m^tw}{n} &\geq \Phi_r\left(\frac{1}{m^t}, \frac{\gamma}{t}\right).
\end{align*}
\end{theorem}

\begin{proof}[Proof of Theorem~\ref{one_probe_thm}]
The lower bound follows from a direct application of Lemma~\ref{lem:robustexpansion} to Theorem~\ref{thm:ptwbound}. Setting $t = 1$ in Theorem~\ref{thm:ptwbound}, we obtain
\[ mw \geq n \cdot \Phi_{r}\left(\frac{1}{m}, \gamma\right) \geq n \gamma^q m^{1 + \frac{q}{p} - q} \]
for some $p, q \in [1, \infty)$ and $(p-1)(q-1) = \sigma^2$. 
Rearranging the inequality and letting $p = 1 + \frac{\log \log n}{\log n}$, and $q = 1 + \sigma^2 \frac{\log n}{\log \log n}$, we obtain
\[ m \geq \dfrac{\gamma^{\frac{p}{p-1}} n^{\frac{p}{pq-q}}}{w^{\frac{p}{pq-q}}} \geq n^{\frac{1}{\sigma^2} - o(1)}. \]
Since $\sigma = 1 - \frac{1}{c}$ and $w = n^{o(1)}$, we obtain the desired result. 
\end{proof}

\begin{corollary}
Any 1 cell probe data structure with cell size $n^{o(1)}$ for $c$-approximate nearest neighbors on the sphere in $\ell_2$ needs $n^{1 + \frac{2c^2 - 1}{(c^2 -1)^2}-o(1)}$ many cells. 
\end{corollary}

\begin{proof}
Each point in the Hamming space $\{-1, 1\}^d$ (after scaling by $\frac{1}{\sqrt{d}}$) can be thought of as lying on the unit sphere. If two points are a distance $r$ apart in the Hamming space, then they are $2\sqrt{r}$ apart on the sphere with $\ell_2$ norm. Therefore a data structure for a $c^2$-approximation on the sphere gives a data structure for a $c$-approximation in the Hamming space.
\end{proof}


\section{Lower bounds: list-of-points data structures}
\label{sec:noCoding}

In this section, we prove Theorem~\ref{thm:noCoding}, i.e., a tight
lower bound against data structure that fall inside the
``list-of-points'' model, as defined in Def.~\ref{def:lip}.

\begin{theorem}[Restatement of Theorem~\ref{thm:noCoding}]
Let $D$ be a list-of-points data structure which solves $(c,r)$-ANN for $n$ points in the $d$-dimensional Hamming space with $d = \omega(\log n)$. Suppose $D$ is specified by a sequence of $m$ sets $\{ A_i \}_{i=1}^m$ and a procedure for outputting a subset $I(q) \subset [m]$ using expected space $s = n^{1+\rho_u}$, and expected query time $n^{\rho_q - o(1)}$ with success probability $\frac{2}{3}$. Then
\[ c \sqrt{\rho_q} + (c-1) \sqrt{\rho_u} \geq \sqrt{2c - 1}. \]
\end{theorem}

We will prove the lower bound by giving a lower bound on list-of-points data structures which solve the random instance for the Hamming space defined in Section~\ref{sec:rand_inst_sec}. The dataset consists of $n$ points $\{ u_i \}_{i=1}^n$ where each  $u_i \sim V$ drawn uniformly at random, and a query $v$ is drawn from the neighborhood of a random dataset point. Thus, we may assume $D$ is a deterministic data structure.

Fix a data structure $D$, where $A_i \subset V$ specifies which dataset points are
placed in $L_i$. Additionally, we may define $B_i \subset V$ which specifies which query points scan $L_i$, i.e., $B_i = \{ v \in V \mid i \in
I(v)\}$. Suppose we sample a random dataset point $u \sim V$ and then
a random query point $v$ from the neighborhood of $u$. Let
\begin{align*}
\gamma_i &= \Pr[v \in B_i \mid u \in A_i]
\end{align*}
represent the probability that query $v$ scans the list $L_i$, conditioned on $u$ being in $L_i$. Additionally, we write $s_i = \mu(A_i)$ as the normalized size of $A$. 
The query time for $D$ is given by the following expression:
\begin{align}
T &= \sum_{i=1}^m \chi_{B_i}(v) \left( 1 + \sum_{j=1}^n \chi_{A_i}(u_j)\right) \nonumber \\
\E[T] &= \sum_{i=1}^m \mu(B_i) + \sum_{i=1}^m \gamma_i \mu(A_i) + (n-1) \sum_{i=1}^m \mu(B_i) \mu(A_i) \nonumber \\
	&\geq \sum_{i=1}^m \Phi_r(s_i, \gamma_i) s_i + \sum_{i=1}^m s_i \gamma_i + (n-1) \sum_{i=1}^m \Phi_r(s_i, \gamma_i) s_i^2. \label{eq:hehe}
\end{align}
Since the data structure succeeds with probability $\gamma$,
\begin{align}
\sum_{i=1}^m s_i \gamma_i &\geq \gamma = \Pr_{j \sim [n], v \sim N(u_j)}[ \exists i \in [m] : v \in B_i , u_j \in A_i] . \label{eq:hehe2}
\end{align}
Additionally, since $D$ uses at most $s$ space, 
\begin{align}
n \sum_{i=1}^m s_i &\leq O(s). \label{eq:hehe3}
\end{align}
Using the two constraints in (\ref{eq:hehe2}) and (\ref{eq:hehe3}), we will use the estimates of robust expansion in order to find a lower bound for (\ref{eq:hehe}). From Lemma~\ref{lem:robustexpansion}, for any $p, q \in [1, \infty)$ where $(p-1)(q-1) = \sigma^2$ where $\sigma = 1 - \frac{1}{c}$,
\begin{align*}
\E[T] &\geq \sum_{i=1}^m s_i^{q - \frac{q}{p}} \gamma_i^q + (n-1) \sum_{i=1}^m s_i^{q - \frac{q}{p} + 1} \gamma_i^q + \gamma \\
\gamma &\leq \sum_{i=1}^m s_i \gamma_i \\
O\left(\frac{s}{n}\right) &\geq \sum_{i=1}^m s_i.
\end{align*}
We set $S = \{i \in [m] : s_i \neq 0 \}$ and for $i \in S$, we write $v_i = s_i \gamma_i$. Then
\begin{equation}
\E[T] \geq \sum_{i \in S} v_i^q \left(s_i^{-\frac{q}{p}} + (n-1)s_i^{-\frac{q}{p} + 1} \right) \geq \sum_{i \in S} \left(\dfrac{\gamma}{|S|} \right)^q \left(s_i^{-\frac{q}{p}} + (n-1)s_i^{-\frac{q}{p} + 1} \right) \label{eq:intermediate}
\end{equation}
where we used the fact $q \geq 1$. Consider
\begin{equation}
\label{eq:F}
F = \sum_{i \in S} \left(s_i^{-\frac{q}{p}}  + (n-1)s_i^{-\frac{q}{p} + 1} \right).
\end{equation}

We analyze three cases separately:
\begin{itemize}
\item $0 < \rho_u \leq \frac{1}{2c - 1}$
\item $\frac{1}{2c - 1} < \rho_u \leq \dfrac{2c-1}{(c-1)^2}$
\item $\rho_u = 0$.  
\end{itemize}
For the first two cases, we let 
\begin{equation}
\label{eq:p-and-q}
q = 1 - \sigma^2 + \sigma \beta \qquad p = \dfrac{\beta}{\beta - \sigma} \qquad \beta = \sqrt{\dfrac{1 - \sigma^2}{\rho_u}} 
\end{equation}
Since $0 < \rho_u \leq \dfrac{2c-1}{(c-1)^2}$, one can verify $\beta > \sigma$ and both $p$ and $q$ are at least $1$. 

\begin{lemma}
\label{lem:q-bigger-than-p}
When $\rho_u \leq \frac{1}{2c - 1}$, and $s = n^{1 + \rho_u}$,
\[ \E[T] \geq \Omega(n^{\rho_q}) \]
where $\rho_q$ and $\rho_u$ satisfy Equation~\ref{eq:power-relation}.
\end{lemma}

\begin{proof}
In this setting, $p$ and $q$ are constants, and $q \geq p$. Therefore, $\frac{q}{p} \geq 1$. $F$ can be viewed as consisting of the contributions of each $s_i$'s in Equation~\ref{eq:F}, constrained by (\ref{eq:hehe3}). One can easily verify that $F$ minimized when $s_i = O(\frac{s}{n|S|})$, so substituting in (\ref{eq:intermediate}),
\[ \E[T] \geq \Omega\left(\dfrac{\gamma^qs^{-q/p + 1}n^{q/p}}{|S|^{q - q/p}} \right) \geq \Omega(\gamma^q s^{1 - q} n^{q/p}) \label{eq:hehe4}\]
since $q - q/p > 0$ and $|S| \leq s$.
In addition, $p$, $q$ and $\gamma$ are constants, and note the fact $s = n^{1+\rho_u}$, and (\ref{eq:p-and-q}), we let $n^{\rho_q}$ be the best query time we can achieve. Combining these facts, along with the lower bound for $\rho_q$ in (\ref{eq:hehe4}), we obtain the following relationship between $\rho_q$ and $\rho_u$:
\begin{align*}
\rho_q &= (1 + \rho_u)(1 - q) + \frac{q}{p} \\
	 &= (1 + \rho_u)(\sigma^2 - \sigma \beta) + \dfrac{(1 - \sigma^2 + \sigma \beta)(\beta - \sigma)}{\beta} \\
	 &= \left( \sqrt{1 - \sigma^2} - \sqrt{\rho_u} \sigma \right)^2 \\
	 &= \left( \dfrac{\sqrt{2c-1}}{c} - \sqrt{\rho_u} \cdot \dfrac{(c-1)}{c} \right)^2.
\end{align*}
\end{proof}

\begin{lemma}
When $\rho_u > \frac{1}{2c - 1}$,
\[ \E[T] \geq \Omega(n^{\rho_q}) \]
where $\rho_q$ and $\rho_u$ satisfy Equation~\ref{eq:power-relation}.
\end{lemma}

\begin{proof}
We follow a similar pattern to Lemma~\ref{lem:q-bigger-than-p}. 
\begin{align*}
\frac{\partial F}{\partial s_i} &= \left(-\frac{q}{p}\right) s_i^{-\frac{q}{p} - 1} + \left(-\frac{q}{p} + 1 \right)(n-1) s_i^{-\frac{q}{p}}.
\end{align*}
Consider the case when each $\frac{\partial F}{\partial s_i}(s_i) = 0$, by setting $s_i = \dfrac{q}{(p - q)(n-1)}$. Since $q < p$, this value is positive and $\sum_{i \in S} s_i \leq O\left(\frac{m}{n}\right)$ for large enough $n$. Thus, $F$ is minimized at this point, and $\E[T] \geq \left(\frac{\gamma}{|S|}\right)^{q} |S| \left( \frac{q}{(p - q)(n-1)}\right)^{-\frac{q}{p}}$. Since $q \geq 1$ and $|S| \leq s$,
\begin{align*}
\E[T] &\geq \left( \frac{\gamma}{s} \right)^{q} s \left(\frac{q}{(p - q)(n-1)}\right)^{-\frac{q}{p}}.
\end{align*}
Since $p$, $q$ and $\gamma$ are constants, $\E[T] \geq \Omega(n^{\rho_q})$,
\[ \rho_q = (1 + \rho_u)(1 - q) + \frac{q}{p} \]
which is the same expression for $\rho_q$ as in Lemma~\ref{lem:q-bigger-than-p}.
\end{proof}

\begin{lemma}
When $\rho_u = 0$ (so $s = O(n)$), 
\[ \E[T] \geq n^{\rho_q - o(1)} \]
where $\rho_q = \dfrac{2c-1}{c^2} = 1 - \sigma^2$.
\end{lemma}

\begin{proof}
In this case, we let
\[ q = 1 + \sigma^2 \cdot \dfrac{\log n}{\log \log n} \qquad p = 1 + \dfrac{\log \log n}{\log n}. \]
Since $q > p$, we have
\[ \E[T] = \Omega(\gamma^q s^{1 - q}n^{\frac{q}{p}}) = n^{1 - \sigma^2 - o(1)}, \]
which is the desired expression. 
\end{proof}


\newcommand{\Cbb}{\mathbb{C}}

\section{Lower bounds: two-probe data structures}
\label{sec:twoProbes}

In this section, we prove the cell probe lower bound for $t=2$ cell probes
 stated in Theorem~\ref{two_probe_thm}.

We follow the framework in \cite{PTW10} and prove lower bounds for GNS when $U = V$ with measure $\mu$ (see Def.~\ref{def:gns}). We assume there is an underlying graph $G$
with vertex set $V$. For any point $p \in V$, we write $p$'s \emph{neighborhood}, $N(p)$, as the set of points with an edge incident on $p$ in $G$.

In the 2-probe GNS problem, we are given a dataset $P = \{ p_i \}_{i=1}^n
\subset V$ of $n$ points as well as a bit-string $x \in \{0,
1\}^n$. The goal is to build a data structure supporting the following types of queries: given a point $q \in V$, if there exists a unique neighbor $p_i \in
N(q) \cap P$, return $x_i$ with probability at least $\frac{2}{3}$ after making two cell-probes.

We let $D$ denote a data structure with $m$ cells of $w$ bits
each. $D$ will depend on the dataset $P$ as well as the bit-string $x$. We will prove the following theorem.

\begin{theorem}
\label{thm:2-query}
There exists a constant $\gamma > 0$ such that any non-adaptive GNS data structure holding a dataset of $n \geq 1$ points which succeeds with probability $\frac{2}{3}$ using two cell probes and $m$ cells of $w$ bits satisfies
\[ \dfrac{m \log m \cdot 2^{O(w)}}{n} \geq \Omega\left(\Phi_r\left(\frac{1}{m}, \gamma\right)\right).  \]
\end{theorem}

Theorem~\ref{two_probe_thm} will follow from Theorem~\ref{thm:2-query}
together with the robust expansion bound from
Lemma~\ref{lem:robustexpansion} for the special case of \emph{non-adaptive} probes. We will later show how to reduce adaptive algorithms losing a sub-polynomial factor in the space for $w = o(\log n)$ in Section~\ref{sec:adaptivity}. We now proceed to proving Theorem~\ref{thm:2-query}.

At a high-level, we show that a ``too-good-to-be-true'', 2-probe data structure implies a weaker notion of 2-query locally-decodable code (LDC) with small noise rate using the same amount of space\footnote{A 2-query LDC corresponds to LDCs which make two probes to their memory contents. Even though there is a slight ambiguity with the data structure notion of query, we say ``2-query LDCs'' in order to be consistent with the LDC literature.}. Even though our notion of LDC is weaker than Def.~\ref{def:ldc}, we adapt the tools for showing 2-query LDC lower bounds from \cite{KW2004}. These arguments, using quantum information theory, are very robust and work well with the weaker 2-query LDC we construct. 

We note that \cite{PTW08} was the first to suggest the connection between ANN and LDCs. This work represents the first concrete connection which gives rise to better lower bounds. 

\paragraph{Proof structure.} 
The proof of Theorem~\ref{thm:2-query} proceeds in six steps.

\begin{enumerate}
\item First we use Yao's principle to focus on deterministic non-adaptive data structures for GNS with two cell-probes. We provide distributions over $n$-point datasets $P$, as well as bit-strings $x$ and a query $q$, and assume the existence of a deterministic data structure succeeding with probability at least $\frac{2}{3}$.
\item We simplify the deterministic data structure in order to get ``low-contention'' data structures. These are data structures which do not rely on any single cell too much (similar to Def. 6.1 in \cite{PTW10}).
\item We use ideas from \cite{PTW10} to understand how queries neighboring particular dataset points probe various cells of the data structure. We fix an $n$-point dataset $P$ with a constant fraction of the points satisfying the following condition: many possible queries in the neighborhood of these points probe disjoint pairs of cells.
\item For the fixed dataset $P$, we show that we can recover a constant fraction of bits of $x$ with significant probability even if we corrupt the contents of some cells. 
\item We reduce to data structures with $1$-bit words in order to apply the LDC arguments from \cite{KW2004}.
\item Finally, we design an LDC with weaker guarantees and use the arguments in \cite{KW2004} to prove lower bounds on the space of the weak LDC.
\end{enumerate}

\subsection{Deterministic data structures}
\label{sec:det-ds}

\begin{definition}
\label{def:rand-gns-ds}
A non-adaptive randomized algorithm $R$ for the GNS problem making two cell-probes is an algorithm specified by the following two components:
\begin{enumerate} 
\item A procedure which preprocess a dataset $P = \{ p_i \}_{i=1}^n$ of $n$ points, as well as a bit-string $x \in \{0, 1\}^n$ in order to output a data structure $D \in \left(\{0, 1\}^w\right)^m$.
\item An algorithm $R$ that given a query $q$, chooses two indices $(i, j) \in [m]^2$ and specifies a function $f_q \colon \{0, 1\}^w \times \{0, 1\}^w \to \{0, 1\}$.
\end{enumerate}
We require the data structure $D$ and the algorithm $R$ satisfy
\[ \Pr_{R, D}[f_q(D_j, D_k) = x_i] \geq \frac{2}{3} \]
whenever $q \in N(p_i)$ and $p_i$ is the unique such neighbor.
\end{definition}

Note that the procedure which outputs the data structure does not depend on the query $q$, and that the algorithm $R$ does not depend on the dataset $P$ or bit-string $x$.

\begin{definition}
We define the following distributions:
\begin{itemize}
\item Let $\Pd$ be the uniform distribution supported on $n$-point datasets from $V$.
\item Let $\X$ be the uniform distribution over $\{0, 1\}^n$.
\item Let $\Q(P)$ be the distribution over queries given by first drawing a dataset point $p \in P$ uniformly at random and then drawing $q \in N(p)$ uniformly at random.
\end{itemize}
\end{definition}

\begin{lemma}
\label{lem:det-alg}
Assume $R$ is a non-adaptive randomized algorithm for GNS using two cell-probes. Then, there exists a non-adaptive deterministic algorithm $A$ for GNS using two cell-probes succeeding with probability at least $\frac{2}{3}$ when the dataset $P \sim \Pd$, the bit-string $x \sim \X$, and $q \sim \Q(P)$.
\end{lemma}


\begin{proof}
We apply Yao's principle to the success probability of the algorithm. By assumption, there exists a distribution over algorithms which can achieve probability of success at least $\frac{2}{3}$ for any single query. Therefore, for the fixed distributions $\Pd, \X,$ and $\Q$, there exists a deterministic algorithm achieving at least the same success probability. 
\end{proof}

In order to simplify notation, we let $A^D(q)$
denote output of the algorithm $A$. We assume that
$A(q)$ outputs a pair of indices $(j, k)$ as well as the function
$f_q \colon \{0, 1\}^w \times \{0, 1\}^w \to \{0, 1\}$, and thus, we use $A^{D}(q)$
as the output of $f_q(D_j, D_k)$. For any fixed dataset $P = \{ p_i \}_{i=1}^n$
and bit-string $x \in \{0, 1\}^n$, 
\[ \Pr_{q \sim N(p_i)}[A^D(q) = x_i] = \Pr_{q \sim N(p_i)}[f_q(D_j, D_k) = x_i]. \]
This notation allows us to succinctly state the probability of correctness when the query is a neighbor of $p_i$.

For the remainder of the section, we let $A$ denote a non-adaptive
deterministic algorithm succeeding with probability at least
$\frac{2}{3}$ using $m$ cells of width $w$. The success probability is
taken over the random choice of the dataset $P \sim \Pd$, $x \sim \X$ and $q \sim \Q(P)$. 

\subsection{Making low-contention data structures}

For any $t \in \{1, 2\}$ and $j \in [m]$, let $A_{t, j}$ be the set of queries which probe cell $j$ at the $t$-th probe of algorithm $A$. 
Since $A$ is deterministic, the indices $(i, j) \in [m]^2$ which $A$ 
outputs are completely determined by two collections $\A_1 = \{ A_{1, j} \}_{j\in [m]}$ and $\A_2 = \{ A_{2, j} \}_{j \in [m]}$ which independently
 partition the query space $V$. On query $q$, if $q \in A_{1, i}$ and $q \in A_{2, j}$,
 algorithm $A$ outputs the indices $(i,j)$. 
 
 We now define the notion of low-contention data structures, which
 requires the data structure to not rely on any one
particular cell too much by ensuring no $A_{t, j}$ is too large.

\begin{definition}
A deterministic non-adaptive algorithm $A$ using $m$ cells has {\em low contention} if every set $\mu(A_{t, j}) \leq \frac{1}{m}$ for $t \in \{1, 2\}$ and $j \in [m]$.   
\end{definition}

We now use the following lemma to argue that up to a small increase in
space, a data structure can be made low-contention. 

\begin{lemma}
\label{lem:heavycells}
Let $A$ be a deterministic non-adaptive algorithm for GNS making two cell-probes using $m$ cells. There exists a deterministic non-adaptive algorithm $A'$ for GNS making two cell-probes using $3m$ cells which has low contention and succeeds with the same probability.
\end{lemma}

\begin{proof}
Suppose $\mu(A_{t, j}) \geq \frac{1}{m}$ for some $j \in [m]$. We partition $A_{t, j}$ into enough sets $\{A^{(j)}_{t, k} \}_k$ of measure $\frac{1}{m}$ and at most one set with measure between $0$ and $\frac{1}{m}$. For each of set $A^{(j)}_{t, k}$, we make a new cell $j_k$ with the same contents as cell $j$. When a query lies inside $A^{(j)}_{t, k}$ the $t$-th probe is made to the new cell $j_k$ instead of cell $j$. 

We apply the above transformation on all sets with $\mu(A_{t, j}) \geq \frac{1}{m}$. In the resulting data structure, in each partition $\A_1$ and $\A_2$, there can be at most $m$ cells of measure $\frac{1}{m}$ and at most $m$ sets with measure less than $\frac{1}{m}$. Therefore, the transformed data structure has at most $3m$ cells. Since the contents remain the same, the data structure succeeds with the same probability.
\end{proof}

Given Lemma~\ref{lem:heavycells}, we assume that $A$ is a deterministic non-adaptive algorithm for GNS with two cell-probes using $m$ cells which has low contention. The extra factor of $3$ in the number of cells is absorbed in the asymptotic notation.

\subsection{Datasets which shatter}

We fix some $\gamma > 0$ to be a sufficiently small constant.

\begin{definition}[Weak-shattering \cite{PTW10}]
We say a partition $A_1, \dots, A_m$ of $V$ $(K, \gamma)$-\emph{weakly shatters a point} $p$ if
\[ \sum_{i \in [m]} \left( \mu(A_i \cap N(p)) - \frac{1}{K}\right)^+ \leq \gamma, \]
where the operator $( \cdot )^+ \colon \R \to \R^+$ is the identity on positive real numbers and zero otherwise.
\end{definition}

\begin{lemma}[Shattering \cite{PTW10}]
\label{lem:shattering}
Let $A_1, \dots, A_k$ collection of disjoint subsets of measure at most $\frac{1}{m}$. Then 
\[ \Pr_{p \sim \mu}[\text{$p$ is $(K, \gamma)$-weakly shattered}] \geq 1 - \gamma \]
for $K = \Phi_r\left(\frac{1}{m}, \frac{\gamma^2}{4}\right) \cdot \frac{\gamma^3}{16}$.
\end{lemma}

For the remainder of the section, we let
\[ K = \Phi_{r}\left( \frac{1}{m}, \frac{\gamma^2}{4} \right) \cdot \frac{\gamma^3}{16}. \]

We are interested in dataset points which are shattered with respect to the collections $\A_1$ and $\A_2$. Intuitively, queries which are near-neighbors of these dataset points probe various disjoint cells in the data structure, so their corresponding bit is stored in many cells. 

\begin{definition}
\label{def:slack}
Let $p \in V$ be a dataset point which is $(K, \gamma)$-weakly shattered by $\A_1$ and $\A_2$. Let $\beta_1, \beta_2 \subset N(p)$ be arbitrary subsets where each $j \in [m]$ satisfies
\[ \mu(A_{1, j} \cap N(p) \setminus \beta_1) \leq \frac{1}{K} \]
and
\[ \mu(A_{2, j} \cap N(p) \setminus \beta_2) \leq \frac{1}{K} \]
Since $p$ is $(K, \gamma)$-weakly shattered, we can pick $\beta_1$ and $\beta_2$ with measure at most $\gamma$ each. We will refer to $\beta(p) = \beta_1 \cup \beta_2$.
\end{definition}

For a fixed dataset point $p \in P$, we refer to $\beta(p)$ as the set holding the \emph{slack} in the shattering of measure at most $2\gamma$. For a given collection $\A$, let $S(\A, p)$ be
 the event that the
collection $\A$ $(K, \gamma)$-weakly shatters $p$. Note that
Lemma~\ref{lem:shattering} implies that $\Pr_{p \sim \mu}[S(\A, p)] \geq 1 - \gamma$. 

\begin{lemma}
\label{lem:gooddbshattered}
With high probability over the choice of $n$-point dataset, at most $4\gamma n$ points do not satisfy $S(\A_1, p)$ and $S(\A_2, p)$. 
\end{lemma}

\begin{proof}
The expected number of points $p$ which do not satisfy $S(\A_1, p)$ and $S(\A_2, p)$ is at most $2 \gamma n$. Therefore via a Chernoff bound, the probability that more than $4\gamma n$ points do not satisfy $S(\A_1, p)$ and $S(\A_2, p)$ is at most $\exp\left(-\frac{2\gamma n}{3}\right)$.
\end{proof}

We call a dataset \emph{good} if there are at most $4\gamma n$ dataset points which are not $(K, \gamma)$-weakly shattered by $\A_1$ and $\A_2$.

\begin{lemma}
\label{lem:good-db}
There exists a good dataset $P = \{ p_i \}_{i=1}^n$ where
\[ \Pr_{x \sim \X, q \sim \Q(P)}[A^D(q) = x_i] \geq \frac{2}{3} - o(1) \]
\end{lemma}

\begin{proof}
 For any fixed dataset $P = \{ p_i \}_{i=1}^n$, let
  $$
  \textbf{P} = \Pr_{x \sim \X, q \sim Q(p)}[A^D(q) = x_i].
  $$ 
  Then,
\begin{align*}
\frac{2}{3} &\leq \mathop{\E}_{P \sim \Pd}[ \textbf{P} ] \\
		&= (1 - o(1)) \cdot \mathop{\E}_{P \sim \Pd} [ \textbf{P} \mid \text{ $P$ is good}] + o(1) \cdot \mathop{\E}_{P\sim \Pd}[\textbf{P} \mid \text{ $P$ is not good}] \\
\frac{2}{3} - o(1) &\leq (1 - o(1)) \cdot \mathop{\E}_{P \sim \Pd}[ \textbf{P} \mid \text{$P$ is good}].
\end{align*}
Therefore, there exists a dataset which is not shattered by at most $4\gamma n$ and $\Pr_{x \sim \X, q \sim \Q(P)}[A^D(y) = x_i] \geq \frac{2}{3} - o(1)$. 
\end{proof}

\subsection{Corrupting some cell contents of shattered points}

In the rest of the proof, we fix the dataset $P = \{ p_i \}_{i=1}^n$ satisfying the
conditions of Lemma~\ref{lem:good-db}, i.e., $P$ satisfies
\begin{align}
\label{eq:condP}
 \Pr_{x \sim \X, q \sim \Q(P)}[ A^D(q) = x_i ] \geq \frac{2}{3} - o(1). 
\end{align}

We now introduce the notion of \emph{corruption} of the data structure cells
$D$, which parallels the notion of noise in locally-decodable codes.
Remember that, after fixing some bit-string $x$, the algorithm $A$
produces some data structure $D \in \left(\{0, 1\}^w\right)^m$. 

\begin{definition}
We call $D' \in \left( \{0, 1\}^w \right)^m$ a \emph{corrupted} version of $D$ at $k$ cells if
$D$ and $D'$ differ on at most $k$ cells, i.e., if $|\{ i
\in [m] : D_i \neq D'_i\}| \leq k$. 
\end{definition}
  
\begin{definition}
For a fixed $x \in \{0,1\}^n$, let 
\begin{align}
\label{eq:c-and-prob} 
c_x(i) = \Pr_{q \sim N(p_i)}[A^D(q) = x_i] 
\end{align}
denote the \emph{recovery probability of bit $i$}.
Note that from the definitions of $\Q(P)$, $\E[c_x(i)] \geq \frac{2}{3} - o(1)$, where the expectation is taken over $x \sim \X$ and $i \in [n]$.
\end{definition}

In this section, we show there exist a subset $S \subset [n]$ of size 
$\Omega(n)$ where each $i \in S$ has constant recovery probability averaged over $x \sim \X$, 
even if the algorithm probes a corrupted version of data structure. We let $\eps > 0$
be a sufficiently small constant.

\begin{lemma}
\label{lem:corrupted}
Fix a vector $x \in \{0, 1\}^n$, and let $D \in \left( \{0, 1\}^w\right)^m$
be the data structure that algorithm $A$ produces on dataset $P$ and
bit-string $x$. Let $D'$ be a corruption of $D$ at $\eps K$ cells. For
every $i\in[n]$ where events $S(\A_1, p_i)$ and $S(\A_2, p_i)$
occur, 
\[ \Pr_{q \sim N(p_i)}[ A^{D'}(q) = x_i] \geq c_x(i) - 2\gamma - 2\eps. \]
\end{lemma}

\begin{proof}
Note that $c_x(i)$ represents the probability that algorithm $A$ run on a uniformly chosen query
from the neighborhood of $p_i$ returns the correct answer, i.e. $x_i$. We denote 
the subset $C_{1} \subset N(p)$ of queries that when run on $A$ return $x_i$; so, $\mu(C_{1}) = c_x(i)$ by definition.

By assumption, $p_i$ is $(K, \gamma)$-weakly shattered by
$\A_1$ and $\A_2$, so by Def.~\ref{def:slack}, we specify some $\beta(p) \subset N(p)$ where $\mu(C_{1} \cap \beta(p)) \leq \mu(\beta(p)) \leq 2\gamma$. 
Let $C_{2} = C_{1} \setminus \beta(p)$, where $\mu(C_2) \geq c_i(x) - 2\gamma$. 
Again, by assumption that $p_i$ is $(K, \gamma)$-weakly shattered, 
each $j \in [m]$ and $t \in \{1, 2\}$ satisfy 
$\mu(C_2 \cap A_{t, j}) \leq \frac{1}{K}$. Let $\Delta \subset [m]$ 
be the set of $\eps K$ cells where $D$ and $D'$ differ, and let
$C_3 \subset C_2$ be given by
\[ C_3 = C_2 \setminus \left(\bigcup_{j \in \Delta} (A_{1, j} \cup A_{2, j}) \right). \]
Thus,
\[ \mu(C_3) \geq \mu(C_2) - \sum_{j \in \Delta} \left(\mu(C_2 \cap A_{1, j}) + \mu(C_2 \cap A_{2, j})\right) \geq c_i(x) - 2\gamma - 2\eps. \]
If $q \in C_3$, then on query $q$, algorithm $A$ probes cells outside of $\Delta$, so $A^{D'}(q) = A^{D}(q) = x_i$. 
\end{proof}

\begin{lemma}
\label{lem:prob-c-x}
There exists a set $S \subset [n]$ of size $\Omega(n)$ with the following property. If $i \in S$, then events $S(\A_1, p_i)$ and $S(\A_2, p_i)$ occur, and
\[ \mathop{\E}_{x \sim \X}[c_x(i)] \geq \frac{1}{2} + \nu, \]
where $\nu$ is a constant. \footnote{One can think of $\nu$ as around $\frac{1}{10}$.}
\end{lemma}

\begin{proof}
For $i \in [n]$, let $E_i$ be the event that $S(\A_1, p_i)$ and $S(\A_2, p_i)$ occur and $\E_{x \sim \X}[c_x(i)] \geq \frac{1}{2} + \nu$. Additionally, let
\[ \textbf{P} = \Pr_{i \in [n]}\left[ E_i \right]. \]
We set $S = \left\{ i \in [n] \mid E_i \right\}$, so it remains to show that $\textbf{P} = \Omega(1)$. 
To this end, 
\begin{align*}
\frac{2}{3} - o(1) &\leq \mathop{\E}_{x \sim \X, i \in [n]}[c_x(i)]  & \text{(by Equations~\ref{eq:condP} and \ref{eq:c-and-prob})}\\
			 &\leq 4\gamma  + \textbf{P} + \left(\frac{1}{2} + \nu\right) \cdot (1 - \textbf{P}) & \text{(since $P$ is good)}\\
\frac{1}{6} - o(1) - 4\gamma - \nu &\leq \textbf{P} \cdot \left( \frac{1}{2} - \nu \right).
\end{align*}
\end{proof}

Fix the set $S \subset [n]$ satisfying the conditions of Lemma~\ref{lem:prob-c-x}.
We combine Lemma~\ref{lem:corrupted} and Lemma~\ref{lem:prob-c-x} to obtain the following condition on the dataset. 

\begin{lemma}
\label{lem:prob-corrupt}
Whenever $i \in S$, 
\[ \mathop{\E}_{x \sim \X}\left[ \Pr_{q \sim N(p_i)}[A^{D'}(q) = x_i] \right] \geq \frac{1}{2} + \eta \]
where $\eta = \nu - 2\gamma - 2\eps$ and $D'$ differs from $D$ in $\eps K$ cells. 
\end{lemma}

\begin{proof}
Whenever $i \in S$, $p_i$ is $(K, \gamma)$-weakly shattered. By Lemma~\ref{lem:prob-c-x}, $A$ outputs $x_i$ with probability $\frac{1}{2} + \nu$ on average when probing the data structure $D$ on input $q \sim N(p_i)$, i.e
\[ \mathop{\E}_{x \sim \X} \left[ \Pr_{q \sim N(p_i)} [A^D(q) = x_i] \right] \geq \frac{1}{2} + \nu. \]
Therefore, from Lemma~\ref{lem:corrupted}, if $A$ probes $D'$ which is a corruption of $D$ in any $\eps K$ cells, $A$ will recover $x_i$ with probability at least $\frac{1}{2} + \nu - 2\gamma - 2\eps$ averaged over all $x \sim \X$ where $q \sim N(p_i)$. In other words,
\[ \mathop{\E}_{x \sim \X}\left[ \Pr_{q \sim N(p_i)}[A^{D'}(q) = x_i] \right] \geq \frac{1}{2} + \nu - 2\gamma - 2\eps. \]
\end{proof}

Summarizing the results of the section, we conclude with the following theorem.

\begin{theorem}
\label{thm:nns-ds}
There exists a two-probe algorithm and a subset $S \subseteq [n]$ of size
$\Omega(n)$, satisfying the following property. When $i \in S$, we can recover $x_i$ with probability at least $\frac{1}{2} + \eta$ over a random choice of $x \sim \X$, even if we probe a corrupted version of the data structure at $\eps K$ cells.
\end{theorem}

\begin{proof}
We describe how one can recover bit $x_i$ from a data structure generated by algorithm $A$.
In order to recover $x_i$, we generate a random query $q \sim N(p_i)$ and probe the data structure at the cells specified by $A$. From Lemma~\ref{lem:prob-corrupt}, there exists a set $S \subset [n]$ of size $\Omega(n)$ for which the described algorithm recovers $x_i$ with probability at least $\frac{1}{2} + \eta$, where the probability is taken on average over all possible $x \in \{0, 1\}^n$. 
\end{proof}

Since we fixed the dataset $P = \{ p_i \}_{i=1}^n$ satisfying the conditions of
Lemma~\ref{lem:good-db}, we will abuse a bit of notation, and refer to algorithm $A$
as the algorithm which recovers bits of $x$ described in Theorem~\ref{thm:nns-ds}.
We say that $x \in \{0, 1\}^n$ is an input to
algorithm $A$ in order to initialize the data structure with dataset $P = \{p_i \}_{i=1}^n$ and $x_i$ is the bit associated with
$p_i$.

\subsection{Decreasing the word size}

In order to apply the lower bounds of 2-query locally-decodable codes, we reduce to the case when the word size $w$ is one bit.

\begin{lemma}
\label{lem:decrease-w}
There exists a deterministic non-adaptive algorithm $A'$ which on input $x \in \{0, 1\}^n$ builds a data structure $D'$ using $m\cdot 2^w$ cells of $1$ bit. For any $i \in S$ as well as any corruption $C$ which differs from $D'$ in at most $\eps K$ cells satisfies
\[ \mathop{\E}_{x \in \{0, 1\}^n}\left[ \Pr_{q \sim N(p_i)} [A'^{C}(q) = x_i] \right] \geq \frac{1}{2} + \frac{\eta}{2^{2w}}. \]
\end{lemma}

\begin{proof}
Given algorithm $A$ which constructs the data structure $D \in \left(\{0, 1\}^w\right)^m$ on input $x \in \{0, 1\}^n$, construct the following data structure $D' \in \left( \{0, 1\} \right)^{m \cdot 2^w}$. For each cell $D_j \in \{0, 1\}^w$, make $2^w$ cells containing all parities of the $w$ bits in $D_j$. This procedure increases the size of the data structure by a factor of $2^w$. 

Fix $i \in S$ and $q \in N(p_i)$ be a query. If the algorithm $A$ produces a function $f_q :\{0, 1\}^w \times \{0, 1\}^w \to \{0, 1\}$ which succeeds with probability at least $\frac{1}{2} + \zeta$ over $x \in \{0, 1\}^n$, then there exists a signed parity on some input bits which equals $f_q$ in at least $\frac{1}{2} + \frac{\zeta}{2^{2w}}$ inputs $x \in \{0, 1\}^n$. Let $S_j$ be the parity of the bits of cell $j$ and $S_k$ be the parity of the bits of cell $k$. Let $f_q': \{0, 1\} \times \{0, 1\} \to \{0, 1\}$ denote the parity or the negation of the parity which equals $f_q$ on $\frac{1}{2} + \frac{\zeta}{2^{2w}}$ possible input strings $x \in \{0, 1\}^n$. 

Algorithm $A'$ will evaluate $f_{q}'$ at the cell containing the parity of the $S_j$ bits in cell $j$ and the parity of $S_k$ bits in cell $k$. Let $I_{S_j}, I_{S_k} \in [m \cdot 2^w]$ be the indices of these cells. If $C'$ is a sequence of $m \cdot 2^w$ cells which differ in $\eps K$ many cells from $D'$, then 
\[ \mathop{\E}_{x \in \{0, 1\}^n} \left[ \Pr_{q \sim N(p_i)} [f_q'(C_{I_{S_j}}, C_{I_{S_k}}) = x_i] \right] \geq \frac{1}{2} + \frac{\eta}{2^{2w}} \]
whenever $i \in S$. 
\end{proof}

For the remainder of the section, we will prove a version of
Theorem~\ref{thm:2-query} for algorithms with $1$-bit words. Given
Lemma~\ref{lem:decrease-w}, we will modify the space to $m \cdot
2^{w}$ and the probability to $\frac{1}{2} + \frac{\eta}{2^{2w}}$ to
obtain the answer. So for the remainder of the section, assume
algorithm $A$ has $1$ bit words.

\subsection{Connection to locally-decodable codes}
\label{sec:conn-ldc}

To complete the proof of Theorem~\ref{thm:2-query}, it remains to prove the
following lemma.

\begin{lemma}
\label{lem:2-query-1-bit}
Let $A$ be a non-adaptive deterministic algorithm which makes $2$ cell
probes to a data structure $D$ of $m$ cells of $1$ bit and recover $x_i$ with probability
$\frac{1}{2} + \eta$ on random input $x \in \{0, 1\}^n$ even after $\eps K$ cells are
corrupted whenever $i\in S$ 
for some fixed $S$ of size $\Omega(n)$. Then the following must hold:
\[ \dfrac{m \log m}{n} \geq \Omega\left(\eps K\eta^2 \right).  \]
\end{lemma}

The proof of the lemma uses \cite{KW2004} and relies heavily on
notions from quantum computing. In particular, quantum information
theory applied to LDC lower bounds.

\subsubsection{Crash course in quantum computing}

We introduce a few concepts from quantum computing that are necessary
in our subsequent arguments. The quantum state of a \emph{qubit} is described by a unit-length vector in
$\Cbb^2$. We write the quantum state as a linear combination of the basis states
$(^1_0) = \ket{0}$ and $(^0_1) = \ket{1}$. The quantum state $\alpha =
(^{\alpha_1}_{\alpha_2})$ can be written
\[ \ket{\alpha} = \alpha_1 \ket{0} + \alpha_2\ket{1}\]
where we refer to $\alpha_1$ and $\alpha_2$ as \emph{amplitudes} and $|\alpha_1|^2 + |\alpha_2|^2 = 1$. The quantum state of an $m$-\emph{qubit system} is a unit vector in the tensor product $\Cbb^2 \otimes \dots \otimes \Cbb^2$ of dimension $2^{m}$. The basis states correspond to all $2^m$ bit-strings of length $m$. For $j \in [2^m]$, we write $\ket{j}$ as the basis state $\ket{j_1} \otimes \ket{j_2} \otimes \dots \otimes \ket{j_m}$ where $j = j_1j_2\dots j_m$ is the binary representation of $j$. We will write the $m$-qubit \emph{quantum state} $\ket{\phi}$ as unit-vector given by linear combination over all $2^m$ basis states. So $\ket{\phi} = \sum_{j \in [2^m]} \phi_j \ket{j}$. As a shorthand, $\bra{\phi}$ corresponds to the conjugate transpose of a quantum state.

A \emph{mixed state} $\{ p_i, \ket{\phi_i} \}$ is a probability distribution over quantum states. In this case, we the quantum system is in state $\ket{\phi_i}$ with probability $p_i$. We represent mixed states by a density matrix $\sum p_i \ket{\phi_i} \bra{\phi_i}$. 

A measurement is given by a family of Hermitian positive semi-definite operators
which sum to the identity operator. Given a quantum state $\ket{\phi}$
and a measurement corresponding to the family of operators $\{M_i^{*}
M_i\}_{i}$, the measurement yields outcome $i$ with probability $\|
M_i \ket{\phi}\|^2$ and results in state $\frac{M_i
  \ket{\phi}}{\|M_i\ket{\phi}\|}$, where the norm $\| \cdot \|$ is
the $\ell_2$ norm. We say the measurement makes the \emph{observation}
$M_i$.

Finally, a quantum algorithm makes a query to some bit-string $y \in
\{0, 1\}^m$ by starting with the state $\ket{c}\ket{j}$ and returning
$(-1)^{c \cdot y_j} \ket{c}\ket{j}$. One can think of $c$ as the control qubit taking values $0$ or $1$; if $c = 0$, the state remains unchanged by the query, and if $c = 1$ the state receives a $(-1)^{y_j}$ in its amplitude. The queries may be made in
superposition to a state, so the state $\sum_{c \in \{0, 1\}, j \in
  [m]} \alpha_{cj} \ket{c}\ket{j}$ becomes $\sum_{c \in \{0, 1\}, j
  \in [m]} (-1)^{c \cdot y_j}\alpha_{cj} \ket{c}\ket{j}$.

\subsubsection{Weak quantum random access codes from GNS algorithms}

\begin{definition}
$C:\{0, 1\}^n \to \{0, 1\}^m$ is a $(2, \delta, \eta)$-LDC if there exists a randomized decoding algorithm making at most $2$ queries to an $m$-bit string $y$ non-adaptively, and for all $x \in \{0, 1\}^n$, $i \in [n]$, and $y \in \{0, 1\}^m$ where $d(y, C(x)) \leq \delta m$, the algorithm can recover $x_i$ from the two queries to $y$ with probability at least $\frac{1}{2} + \eta$. 
\end{definition}

In their paper, \cite{KW2004} prove the following result about 2-query LDCs. 

\begin{theorem}[Theorem 4 in \cite{KW2004}]
\label{thm:qldc}
If $C : \{0, 1\}^n \to \{0, 1\}^m$ is a $(2, \delta, \eta)$-LDC, then $m \geq 2^{\Omega(\delta \eta^2 n)}$. 
\end{theorem}

The proof of Theorem~\ref{thm:qldc} proceeds as follows.
They show how to construct a $1$-query quantum-LDC from a classical
$2$-query LDC. From a $1$-query quantum-LDC, \cite{KW2004} constructs
a quantum random access code which encodes $n$-bit strings in $O(\log
m)$ qubits. Then they apply a quantum information theory lower bound
due to Nayak \cite{N1999}:

\begin{theorem}[Theorem 2 stated in \cite{KW2004} from Nayak \cite{N1999}]
\label{thm:nayak}
For any encoding $x \to \rho_x$ of $n$-bit strings into $m$-qubit
states, such that a quantum algorithm, given query access to $\rho_x$,
can decode any fixed $x_i$ with probability at least $1/2+\eta$, it must
hold that $m \geq (1 - H(1/2+\eta)) n$.
\end{theorem}

Our proof will follow a pattern similar to the proof of Theorem
\ref{thm:qldc}. We assume the existence of a GNS algorithm $A$ which
builds a data structure $D \in \{0, 1\}^m$. 

Our algorithm $A$ from Theorem~\ref{thm:nns-ds} does not satisfy the
strong properties of an LDC, preventing us from applying
\ref{thm:qldc} directly. However, it does have some LDC-\emph{ish}
guarantees. In particular, we can recover bits in the presence of $\eps K$ corruptions to
$D$. In the LDC language, this means that we can tolerate a noise
rate of $\delta = \frac{\eps K}{m}$. Additionally, we
cannot necessarily recover \emph{every} coordinate $x_i$, but we can
recover $x_i$ for $i \in S$, where $|S| = \Omega(n)$. Also, our
success probability is $\frac{1}{2} + \eta$ over the random choice of
$i \in S$ and the random choice of the bit-string $x \in \{0,
1\}^n$. Our proof follows by adapting the arguments of \cite{KW2004}
to this weaker setting.

\begin{lemma}
\label{lem:quantum-alg}
Let $r = \frac{2}{\delta a^2}$ where $\delta = \dfrac{\eps K}{m}$ and
$a\le1$ is a constant. Let $D$ be the data structure from
above (i.e., satisfying the hypothesis of Lemma~\ref{lem:2-query-1-bit}). Then there exists a quantum algorithm that, starting from the
$r(\log m + 1)$-qubit state with $r$ copies of $\ket{U(x)}$, where
\[ \ket{U(x)} = \frac{1}{\sqrt{2m}}\sum_{c \in \{0, 1\}, j \in [m]} (-1)^{c \cdot D_j} \ket{c}\ket{j} \]
can recover $x_i$ for any $i \in S$ with probability $\frac{1}{2} +
\Omega(\eta)$ (over a random choice of $x$).  
\end{lemma}

Assuming Lemma~\ref{lem:quantum-alg}, we can complete the proof of
Lemma~\ref{lem:2-query-1-bit}.

\begin{proof}[Proof of Lemma~\ref{lem:2-query-1-bit}]
The proof is similar to the proof of Theorem 2 of \cite{KW2004}.
Let $\rho_x$ represent the $s$-qubit system consisting of the $r$
copies of the state $\ket{U(x)}$, where $s = r(\log m + 1)$; $\rho_x$
is an encoding of $x$.  Using Lemma~\ref{lem:quantum-alg}, we can
assume we have a quantum algorithm that, given
$\rho_x$, can recover $x_i$ for any $i \in S$ with
probability $\alpha = \frac{1}{2} + \Omega(\eta)$ over the random
choice of $x\in \{0,1\}^n$.

We will let $H(A)$ be the Von Neumann entropy of $A$, and $H(A|B)$ be
the conditional entropy and $H(A:B)$ the mutual information.

Let $XM$ be the $(n + s)$-qubit system
\[ \frac{1}{2^n} \sum_{x \in \{0, 1\}^n} \ket{x}\bra{x} \otimes \rho_x. \]
The system corresponds to the uniform superposition of all $2^n$ strings concatenated with their encoding $\rho_x$.
Let $X$ be the first subsystem corresponding to the first $n$ qubits and $M$ be the second subsystem corresponding to the $s$ qubits. We have
\begin{align*} 
H(XM) &= n + \frac{1}{2^n} \sum_{x \in \{0, 1\}^n} H(\rho_x) \geq n = H(X) \\
H(M) &\leq s,
\end{align*}
since $M$ has $s$ qubits. Therefore, the mutual information $H(X : M) = H(X) + H(M) - H(XM) \leq s$. Note that $H(X | M) \leq \sum_{i=1}^n H(X_i | M)$. By Fano's inequality, if $i \in S$, 
\[ H(X_i | M) \leq H(\alpha) \]
where we are using the fact that Fano's inequality works even if we can recover $x_i$ with probability $\alpha$ averaged over all $x$'s.
Additionally, if $i \notin S$, $H(X_i | M) \leq 1$. Therefore,
\begin{align*}
s \geq H(X : M) &= H(X) - H(X|M) \\
  			 &\geq H(X) - \sum_{i=1}^n H(X_i | M) \\
			 &\geq n - |S| H(\alpha) - (n - |S|) \\
			 &= |S| (1 - H(\alpha)).
\end{align*}
Furthermore, $1 - H(\alpha) \geq \Omega(\eta^2)$ since, and $|S| = \Omega(n)$, we have
\begin{align*}
\frac{2m}{a^2\eps K} (\log m + 1) &\geq \Omega\left(n \eta^2\right) \\
\dfrac{m \log m}{n} &\geq \Omega\left(\eps K\eta^2\right).
\end{align*}
\end{proof}

It remains to prove Lemma~\ref{lem:quantum-alg}, which we proceed to
do in the rest of the section.
We first show that we can simulate our GNS algorithm with a 1-query
quantum algorithm.

\begin{lemma}
\label{lem:q-simul}
Fix an $x \in \{0, 1\}^n$ and $i \in [n]$. Let $D \in \{0, 1\}^m$
be the data structure produced by algorithm $A$ on input $x$. Suppose
$\Pr_{q \sim N(p_i)}[A^D(q) = x_i] = \frac{1}{2} + b$ for $b>0$. Then
there exists a quantum algorithm which makes one quantum query (to
$D$) and succeeds with probability $\frac{1}{2} + \frac{4b}{7}$ to
output $x_i$.
\end{lemma}

\begin{proof}
We use the procedure in Lemma 1 of \cite{KW2004} to determine the
output algorithm $A$ on input $x$ at index $i$. The procedure
simulates two classical queries with one quantum query.
\end{proof}

Without loss of generality, all quantum algorithms which make 1-query to $D$ can be specified in
the following manner: there is a quantum state $\ket{Q_i}$, where
\[ \ket{Q_i} = \sum_{c\in\{0, 1\}, j \in [m]} \alpha_{cj}\ket{c}\ket{j} \]
which queries $D$. After querying $D$, the resulting quantum state is
$\ket{Q_i(x)}$, where
\[ \ket{Q_i(x)} = \sum_{c\in \{0, 1\}, j \in [m]} (-1)^{c \cdot D_j} \alpha_{cj}\ket{c}\ket{j}.\]
There is also a quantum measurement $\{ R, I - R \}$ such that, after the
algorithm obtains the state $\ket{Q_i(x)}$, it performs the
measurement $\{ R, I - R\}$. If the algorithm observes $R$, it outputs
$1$ and if the algorithm observes $I - R$, it outputs 0.

From Lemma~\ref{lem:q-simul}, we know there exist a state $\ket{Q_i}$ and a measurement $\{R, I - R\}$ where if algorithm $A$ succeeds with probability $\frac{1}{2} + \eta$ on random $x \sim \{0, 1\}^n$, then the quantum algorithm succeeds with probability $\frac{1}{2} + \frac{4\eta}{7}$ on random $x \sim \{0, 1\}^n$.

In order to simplify notation, we write $p(\phi)$ as the probability of making observation $R$ from state $\ket{\phi}$. Since $R$ is a positive semi-definite matrix, $R = M^{*}M$ and so $p(\phi) = \| M\ket{\phi} \|^2$. 

In exactly the same way as \cite{KW2004}, we can remove parts of the quantum state $\ket{Q_i(x)}$ where $\alpha_{cj} > \frac{1}{\sqrt{\delta m}} = \frac{1}{\sqrt{\eps K}}$. If we let $L = \{ (c, j) \mid \alpha_{cj} \leq \frac{1}{\sqrt{\eps K}} \}$, after keeping only the amplitudes in $L$, we obtain the quantum state $\frac{1}{a}\ket{A_i(x)}$, where
\[ \ket{A_i(x)} = \sum_{(c,j) \in L} (-1)^{c \cdot D_j}\alpha_{cj}\ket{c} \ket{j}, \qquad a = \sqrt{\sum_{(c, j) \in L} \alpha_{cj}^2}. \]

\begin{lemma}
\label{lem:prob-gap}
Fix $i \in S$. The quantum state $\ket{A_i(x)}$ satisfies
\[ \mathop{\E}_{x \in \{0, 1\}^n}\left[ p\left(\frac{1}{a}A_i(x) \right) \mid x_i = 1\right] - \mathop{\E}_{x \in \{0, 1\}^n}\left[ p\left(\frac{1}{a}A_i(x)\right) \mid x_i = 0\right] \geq \frac{8\eta}{7a^2}. \] 
\end{lemma}

\begin{proof}
Note that since $\ket{Q_i(x)}$ and $\{ R , I - R \}$ simulate $A$ and
succeed with probability at least $\frac{1}{2} + \frac{4\eta}{7}$ on a
random $x \in \{0, 1\}^n$, we have that
\begin{align*}
\frac{1}{2}\mathop{\E}_{x \in \{0, 1\}^n}\left[ p\left( Q_i(x) \right) \mid x_i = 1\right] + \frac{1}{2} \mathop{\E}_{x \in \{0, 1\}^n}\left[ 1 - p\left(Q_i(x)\right) \mid x_i = 0 \right] &\geq \frac{1}{2} + \frac{4\eta}{7},
\end{align*} 
which we can simplify to say
\begin{align*}
\mathop{\E}_{x \in \{0, 1\}^n}\left[ p\left( Q_i(x) \right) \mid x_i = 1\right] + \mathop{\E}_{x \in \{0, 1\}^n} \left[ p\left(Q_i(x)\right) \mid x_i = 0 \right] &\geq \frac{8\eta}{7}. 
\end{align*}

Since $\ket{Q_i(x)} = \ket{A_i(x)} + \ket{B_i(x)}$ and $\ket{B_i(x)}$
contains at most $\eps K$ parts, if all probes to $D$ in $\ket{B_i(x)}$
had corrupted values, the algorithm should still succeed with the same
probability on random inputs $x$. Therefore, the following two
inequalities hold:
\begin{align}
\mathop{\E}_{x \in \{0, 1\}^n}\left[ p\left( A_i(x) + B(x) \right) \mid x_i = 1\right] + \mathop{\E}_{x \in \{0, 1\}^n} \left[ p\left(A_i(x) + B(x) \right) \mid x_i = 0 \right] &\geq \frac{8\eta}{7} \label{eq:ineq-1}\\
\mathop{\E}_{x \in \{0, 1\}^n}\left[ p\left( A_i(x) - B(x) \right) \mid x_i = 1\right] + \mathop{\E}_{x \in \{0, 1\}^n} \left[ p\left(A_i(x) - B(x)\right) \mid x_i = 0 \right] &\geq \frac{8\eta}{7} \label{eq:ineq-2}
\end{align}
Note that $p(\phi \pm \psi) = p(\phi) + p(\psi) \pm \left(  \bra{\phi} R \ket{\psi} + \bra{\psi} D \ket{\phi}\right)$ and $p(\frac{1}{c} \phi) = \frac{p(\phi)}{c^2}$. One can verify by averaging the two inequalities (\ref{eq:ineq-1}) and (\ref{eq:ineq-2}) that we get the desired expression. 
\end{proof}

\begin{lemma}
\label{lem:q-alg-ai}
Fix $i \in S$. There exists a quantum algorithm that starting from the quantum state $\frac{1}{a}\ket{A_i(x)}$, can recover the value of $x_i$ with probability $\frac{1}{2} + \frac{2\eta}{7a^2}$ over random $x \in \{0, 1\}^n$. 
\end{lemma}

\begin{proof}
The algorithm and argument are almost identical to Theorem 3 in \cite{KW2004}, we just check that it works under the weaker assumptions. Let 
\[ q_1 = \mathop{\E}_{x \in \{0, 1\}^n}\left[p\left(\frac{1}{a}A_i(x)\right) \mid x_i = 1\right] \qquad q_0 = \mathop{\E}_{x \in \{0, 1\}^n}\left[ p\left(\frac{1}{a}A_i(x)\right) \mid x_i = 0 \right]. \]
From Lemma~\ref{lem:prob-gap}, we know $q_1 - q_0 \geq \frac{8\eta}{7a^2}$. In order to simplify notation, let $b = \frac{4\eta}{7a^2}$. So we want a quantum algorithm which starting from state $\frac{1}{a} \ket{A_i(x)}$ can recover $x_i$ with probability $\frac{1}{2} + \frac{b}{2}$ on random $x \in \{0, 1\}^n$. Assume $q_1 \geq \frac{1}{2} + b$, since otherwise $q_0 \leq \frac{1}{2} - b$ and the same argument will work for $0$ and $1$ flipped. Also, assume $q_1 + q_0 \geq 1$, since otherwise simply outputting $1$ on observation $R$ and $0$ on observation $I - R$ will work. 

The algorithm works in the following way: it outputs $0$ with probability $1 - \frac{1}{q_1 + q_0}$ and otherwise makes the measurement $\{R, I - R\}$ on state $\frac{1}{a} \ket{A_i(x)}$. If the observation made is $R$, then the algorithm outputs $1$, otherwise, it outputs $0$. The probability of success over random input $x \in \{0, 1\}^n$ is 
\begin{multline}
\mathop{\E}_{x \in \{0, 1\}^n} \left[ \Pr[\text{returns correctly}] \right] \\= \frac{1}{2}  \mathop{\E}_{x \in \{0, 1\}^n} \left[ \Pr[\text{returns 1}] \mid x_i = 1 \right] + \frac{1}{2}  \mathop{\E}_{x \in \{0, 1\}^n}\left[ \Pr[\text{returns 0}] \mid x_i = 0 \right]. \label{eq:correct}
\end{multline}
When $x_i = 1$, the probability the algorithm returns correctly is $(1 - q) p\left(\frac{1}{a} A_i(x)\right)$ and when $x_i = 0$, the probability the algorithm returns correctly is $q + (1 - q)(1 - p(\frac{1}{a}A_i(x)))$. So simplifying (\ref{eq:correct}),
\begin{align*}
\mathop{\E}_{x \in \{0, 1\}^n}\left[\Pr[\text{returns correctly}] \right] &= \frac{1}{2}(1 - q) q_1 + \frac{1}{2}(q + (1 - q)(1 - q_0)) \geq \frac{1}{2} + \frac{b}{2}.
\end{align*}
\end{proof}

Now we can finally complete the proof of Lemma~\ref{lem:quantum-alg}.

\begin{proof}[Proof of Lemma~\ref{lem:quantum-alg}]
Again, the proof is exactly the same as the finishing arguments of
Theorem 3 in \cite{KW2004}, and we simply check the weaker conditions
give the desired outcome.  On input $i \in [n]$ and access to $r$
copies of the state $\ket{U(x)}$, the algorithm applies the
measurement $\{ M_i^{*} M_i, I - M_{i}^* M_i \}$ where
\[ M_{i} = \sqrt{\eps K} \sum_{(c ,j) \in L} \alpha_{cj} \ket{c, j}\bra{c , j}. \]

This measurement is designed in order to yield the state $\frac{1}{a}
\ket{A_i(x)}$ on $\ket{U(x)}$ if the measurement makes the observation
$M_i^* M_i$. The fact that the amplitudes  of $\ket{A_i(x)}$ are not
too large makes $\{ M_i^*M_i, I - M_i^*M_i\}$ a valid measurement.

The probability of observing $M_i^{*}M_i$ is $\bra{U(x)} M_i^* M_i
\ket{U(x)} = \frac{\delta a^2}{2}$, where we used that $\delta =
\frac{\eps K}{m}$. So the algorithm repeatedly applies the measurement
until observing outcome $M_i^* M_i$. If it never makes the
observation, the algorithm outputs $0$ or $1$ uniformly at random. If
the algorithm does observe $M_i^* M_i$, it runs the output of the
algorithm of Lemma~\ref{lem:q-alg-ai}. The following simple
calculation (done in \cite{KW2004}) gives the desired probability of
success on random input,
\begin{align*}
\mathop{\E}_{x \in \{0, 1\}^n}\left[ \Pr[\text{returns correctly}] \right] &\geq \left(1 - (1 - \delta a^2/2)^r\right) \left(\frac{1}{2} + \frac{2\eta}{7a^2}\right) + (1 - \delta a^2/2)^r \cdot \frac{1}{2} \geq \frac{1}{2} + \frac{\eta}{7a^2}.
\end{align*}
\end{proof}

\subsubsection{On adaptivity}
\label{sec:adaptivity}

We can extend our lower bounds from the non-adaptive to the adaptive setting. 

\begin{lemma}
\label{lem:adaptivity}
If there exists a deterministic data structure which makes two queries
adaptively and succeeds with probability at least $\frac{1}{2} +
\eta$, there exists a deterministic data structure which makes the two
queries non-adaptively and succeeds with probability at least
$\frac{1}{2} + \frac{\eta}{2^{w}}$.
\end{lemma}

\begin{proof}
The algorithm guesses the outcome of the first cell probe and simulates
the adaptive algorithm with the guess. After knowing which two probes
to make, we probe the data structure non-adaptively. If the algorithm
guessed the contents of the first cell-probe correctly, then we output
the value of the non-adaptive algorithm. Otherwise, we output a random
value. This algorithm is non-adaptive and succeeds with probability at
least $\left(1 - \frac{1}{2^w}\right) \cdot \frac{1}{2} +
\frac{1}{2^w} \left( \frac{1}{2} + \eta \right) = \frac{1}{2} +
\frac{\eta}{2^{w}}$.
\end{proof}

Applying Lemma~\ref{lem:adaptivity}, from an adaptive algorithm succeeding with probability $\frac{2}{3}$, we obtain a non-adaptive algorithm succeeding with probability $\frac{1}{2} + \Omega(2^{-w})$. This value is lower than the intended $\frac{2}{3}$, but we may still reduce to a weak LDC, where we require $\gamma = \Theta(2^{-w})$, $\eps = \Theta(2^{-w})$, and $|S| = \Omega(2^{-w} n)$. With these minor changes to the parameters in Subsections~\ref{sec:det-ds} through~\ref{sec:conn-ldc}, one can easily verify 
\[ \dfrac{m \log m \cdot 2^{\Theta(w)}}{n} \geq \Omega\left( \Phi_r\left( \frac{1}{m}, \gamma\right)\right). \]
This inequality yields tight lower bounds (up to sub-polynomial factors) for the Hamming space when $w = o(\log n)$. 

In the case of the Hamming space, we can compute robust expansion in a similar fashion to Theorem~\ref{one_probe_thm}. In particular, for any $p, q \in [1, \infty)$ where $(p-1)(q-1) = \sigma^2$, we have
\begin{align*}
\dfrac{m \log m \cdot 2^{O(w)}}{n} &\geq \Omega(\gamma^q m^{1 + q/p - q}) \\
m^{q - q/p + o(1)} &\geq n^{1 - o(1)} \gamma^q \\
m &\geq n^{\frac{1 - o(1)}{q - q/p + o(1)}} \gamma^{\frac{q}{q - q/p + o(1)}} = n^{\frac{p}{pq - q} - o(1)} \gamma^{\frac{p}{p - 1} - o(1)}.
\end{align*}
Let $p = 1 + \frac{w f(n)}{\log n}$ and $q = 1 + \sigma^2 \frac{\log n}{w f(n)}$ where we require that $w f(n) = o(\log n)$ and $f(n) \rightarrow \infty$ as $n \rightarrow \infty$. Then,
\begin{align*}
m &\geq n^{\frac{1}{\sigma^2} - o(1)} 2^{\frac{\log n}{wf(n)}} \geq n^{\frac{1}{\sigma^2} - o(1)}.
\end{align*}

\section{Acknowledgments}

We would like to thank Jop Bri\"{e}t for helping us to navigate
literature about LDCs. We also thank Omri Weinstein for useful
discussions. Thanks to Adam Bouland for educating us on the~topic of quantum computing. Thijs Laarhoven is supported by the SNSF ERC Transfer Grant CRETP2-166734 FELICITY.
This material is based upon work supported by the National Science Foundation Graduate
Research Fellowship under Grant No. DGE-16-44869. It is also supported
in part by NSF CCF-1617955 and Google Research Award.

{
\small
\singlespacing
\bibliographystyle{alpha}
\bibliography{bibfile}
}

\end{document}